\documentclass[12pt]{article}

\newcommand{\blind}{0}

\usepackage{graphicx}
\usepackage{amssymb, amsmath, amsthm, bm, bbm}
\usepackage{mathrsfs}
\usepackage{natbib}
\usepackage{float}

\usepackage[title]{appendix}

\usepackage{hyperref}
\hypersetup{
  colorlinks,
  linkcolor = {blue},
  citecolor = {blue},
  urlcolor = {blue}
}

\newtheorem{lemma}{Lemma}
\newtheorem{corollary}{Corollary}

\newcommand{\mat}[1]{\bm{#1}} 

\newcommand{\thetaL}{\hat{\theta}_L}
\newcommand{\thetaU}{\hat{\theta}_U}
\newcommand{\deltaL}{\delta_L}
\newcommand{\deltaU}{\delta_U}
\newcommand{\tran}{\mathsf{T}}

\newcommand{\hthetanj}{\hat{\theta}_j}
\newcommand{\hthetarepj}{\hat{\theta}^{\text{rep}}_j}

\newcommand{\bhthetan}{\bm{\hat{\theta}}}
\newcommand{\bhthetarep}{\bm{\hat{\theta}}^{\text{rep}}}

\newcommand{\thetan}{\hat{\theta}}
\newcommand{\thetarep}{\hat{\theta}^{\text{rep}}}
\newcommand{\thetarepm}{\hat{\theta}^{\text{rep}, m}}

\newcommand{\hSigman}{\mat{\hat{\Sigma}}}
\newcommand{\hsigman}{\hat{\sigma}}
\newcommand{\hsigmasqn}{\hat{\sigma}^{2}}
\newcommand{\hsigmasqnj}{\hat{\sigma}^{2}_j}
\newcommand{\hsigmanj}{\hat{\sigma}_j}

\newcommand{\tdist}{F}

\newcommand{\hnusq}{\hat{\nu}^2}
\newcommand{\Dn}{\mat{D}}
\newcommand{\Dm}{\mat{D}^{\text{rep}}}

\begin{document}

\def\spacingset#1{\renewcommand{\baselinestretch}%
{#1}\small\normalsize} \spacingset{1}


\if0\blind
{
  \title{\bf Towards replicability with confidence intervals for the exceedance probability}
  \author{Brian D. Segal\hspace{.2cm}\\
    Flatiron Health}
  \maketitle
} \fi

\if1\blind
{
  \bigskip
  \bigskip
  \bigskip
  \begin{center}
    {\LARGE\bf Towards replicability with confidence intervals for the exceedance probability}
\end{center}
  \medskip
} \fi

\bigskip

\begin{abstract}
Several scientific fields including psychology are undergoing a replication crisis. There are many reasons for this problem, one of which is a misuse of p-values. There are several alternatives to p-values, and in this paper we describe a complement that is geared towards replication. In particular, we focus on confidence intervals for the probability that a parameter estimate will exceed a specified value in an exact replication study. These intervals convey uncertainty in a way that p-values and standard confidence intervals do not, and can help researchers to draw sounder scientific conclusions. After briefly reviewing background on p-values and a few alternatives, we describe our approach and provide examples with simulated and real data. For linear models, we also describe how confidence intervals for the exceedance probability are related to p-values and confidence intervals for parameters.
\end{abstract}

\noindent%
{\it Keywords:} p-values, reproducibility probability, Bayes factors, psychology
\vfill

\newpage
\spacingset{1.45} 

\section{Introduction}

Several scientific fields including psychology are undergoing a replication crisis \citep{ioannidis2005, prinz2011, begley2012, pashler2012}, which has led to increased attempts to understand the replicability and reliability of scientific studies \citep{open2015, johnson2017}. There are many reasons for this problem, including publication bias, p-hacking, failing to correct for multiple testing, and underpowered studies, though p-values and null hypothesis testing have come under particular scrutiny. This prompted the ASA to release a statement providing guidance on the proper interpretation and use of p-values \citep{wasserstein2016}, as well as a follow-up in which several authors provided additional perspectives and proposal for ``moving to a world beyond `$p<0.05$'\thinspace'' \citep{wasserstein2019}. 

In this article, we discuss a complement to p-values that is similar to the reproducibility probability sometimes used in clinical trials \citep{goodman1992, shao2002, deMartini2008, deMartini2012, deMartini2013, boos2011}, as well as to metrics used in quality control to assess the probability of meeting specifications \citep{meeker2017}. In particular, we focus on confidence intervals for the probability that in an exact replication study, a parameter estimate will exceed a specified value based on sampling variability alone. We refer to these intervals as confidence intervals for the exceedance probability.

Confidence intervals for the exceedance probability convey uncertainty in point estimates in a way that p-values and confidence intervals for parameter estimates (standard confidence intervals) do not. In particular, the exceedance probability can be interpreted similarly to statistical power, and confidence intervals for the exceedance probability convey the stability of a result. Even if an exact replication study is not feasible, we think that characterizing uncertainty in these terms can help researchers to contextualize uncertainty and draw sounder scientific conclusions, particularly with small to medium sample sizes. Like the p-values and standard confidence intervals they complement, confidence intervals for the exceedance probability provide information about the strength of evidence, but do not provide conclusive evidence for whether a result will replicate and are not a substitute for conducting replication studies.

Problems with replication have been prevalent in psychology, among other fields. In psychology there have been efforts to clarify the generalizability of studies \citep{brandt2014}, as well as the criteria for conducting replications \citep{simons2017} and declaring success \citep{rosenthal1997}. However, exact replications are not usually feasible \citep{stroebe2014, fabrigar2016}, which leads to between-study heterogeneity \citep{mcshane2014, vanerp2017, stanley2018, mcshane2019large}. Consequently, confidence intervals for the exceendance probability, which account for sampling variability but not between-study heterogeneity, represent a lower bound on uncertainty. However, as we show through simulations and an analysis of data collected as part of the Open Science Collaboration \citep{open2015}, even this lower bound can help to put p-values into perspective.

Confidence intervals for the exceedance probability have been used in practice in fields other than psychology. In particular, \citet[][Chapters 2.2.3 and 4.5]{meeker2017} gives examples of how confidence intervals for the exceedance probability can be used for quality control to assess the probability of meeting specifications, and \citet{shao2002} and \citet{deMartini2012, deMartini2013} demonstrate how confidence intervals for the reproducibility probability (i.e. the power of an exact replication study) can be used to make decisions in clinical trials regarding sample size and the strength of evidence for single study approvals.

We extend previous work to general linear models and the context of psychological research, and study the relationship with p-values and standard confidence intervals. For linear combinations of normal random variables, there is a clear connection between confidence intervals for the exceedance probability, standard confidence intervals, and p-values, which complements the findings of \citet{deMartini2008}. In addition to providing context for interpreting confidence intervals for the exceedance probability, these relationships provide perspective on p-values and standard confidence intervals.

The work of \citet{gelman2014} is also related, though whereas \citet{gelman2014} calculate the probability of sign and magnitude errors for fixed effect size and variance, we focus on confidence intervals that treat the estimated effect size and variance as random. In a similar vein, \citet{cumming2006} study the probability that a confidence interval from an initial study will contain the point estimate from a replication study. However, whereas \citet{cumming2006} studied these properties in aggregate and used their findings to address common misconceptions of confidence intervals and replicability, we focus on confidence intervals that can be used in individual data analyses.

In Section \ref{background}, we give an overview of common shortcomings of p-values and statistical significance thresholds, and note two prominent suggestions for addressing those issues within a hypothesis testing framework. In Section \ref{motivation_framework}, we give motivating use cases in which the exceedance probability is relevant to the scientific question, and introduce our framework for computing exceedance probabilities and associated confidence intervals. In Section \ref{linComboNorm}, we focus on the exceedance probability for linear combinations of normal random variables and show how confidence intervals for the exceedance probability are related to standard confidence intervals and p-values in this setting. In Section \ref{examples}, we give examples with both simulated and real data of how confidence intervals for the exceedance probability can be used in practice, and how they compare to p-values, standard confidence intervals, and Bayes factors. In Section \ref{conclusion}, we conclude and suggest areas for future work.

\section{Background \label{background}}

\subsection{Limitations of p-values and statistical significance thresholds \label{context}}

Let $\bm{y} \in \mathbb{R}^n$ be an observation of the random vector $\bm{Y}$ that follows a distribution with parameter $\theta \in \mathbb{R}$. Also let $T(\bm{Y}) \in \mathbb{R}$ be a test statistic for which larger values are more extreme, and let $H_0: \theta \in \Theta_0$ and $H_1: \theta \not \in \Theta_0$ be the null and alternative hypotheses. The p-value is given by $p(\bm{y}) = \sup_{\theta \in \Theta_0}\Pr(T(\bm{Y}) \ge T(\bm{y}))$. For example, suppose $Y_i \overset{\text{i.i.d.}}{\sim} N(\theta, \sigma^2)$, $i=1,\ldots,n$, with known variance $\sigma^2$. Then under the null $H_0: \theta = \theta_0$ and alternative $H_1: \theta \ne \theta_0$, we can use the statistic $T(\bm{y}) = \sqrt{n}|\bar{y} - \theta_0| / \sigma$ to obtain the two-sided p-value $p(\bm{y}) = 2[1 - \Phi(T(\bm{y}))]$ where $\bar{y}$ is the sample mean and $\Phi$ is the standard normal cumulative distribution function (CDF).

P-values are simple, scalable summaries of data that can be useful in scientific research if used appropriately. However, p-values are frequently misinterpreted. For example, p-values are commonly interpreted as one minus the probability of replication, or as the posterior probability of the null hypothesis $\Pr(H_0 | \bm{y})$ \citep{oakes1986, cohen1994, haller2002, gigerenzer2018}. As noted by several authors, sometimes the p-value is similar to the posterior probability, but in many cases it is not \citep{lindley1957, pratt1965, berger1987, casella1987, nuzzo2014}. Furthermore, even if interpreted correctly, p-values are random variables that can exhibit large amounts of variability \citep{boos2011}, which can affect the probability of replicating a small p-value in future studies.

The problems with p-values are compounded when they are used in conjunction with dichotomizing statistical significance thresholds \citep{cohen1994, gigerenzer2004, amrhein2017earth, gigerenzer2018}. The dichotomization of p-values based on whether they are above or below a threshold frequently causes researchers to  misinterpret the evidence itself as dichotomous \citep{mcshane2016, mcshane2017}, because statistical significance is commonly confused with scientific significance \citep{mccloskey1996, mcshane2017}. Furthermore, the conventional cutoff of 0.05 is based on historical convenience (being the p-value corresponding to about two standard normal deviations) as opposed to a scientific rationale \citep{kennedy2019, ruberg2019}.

Due to these concerns, many researchers have advocated for reforms, such as shifting emphasis away from p-values and towards other aspects of a study, including how the study was conducted, any problems that arose, and what analytic methods were used, as well as avoiding overconfidence in statistical inference \citep{mcshane2019, amrhein2019inferential}. Others have noted the importance of understanding common cognitive biases, prespecifying analyses, changing institutional norms, focusing on false positive rates, and the use of Bayesian modeling of interactions \citep{gelman2015, leek2017, greenland2017}. There have also been efforts to ban p-values, but these initiatives may have harmed rather than improved the quality of scientific research \citep{fricker2019}.

\subsection{Hypothesis testing alternatives}

Several suggestions have been made to alleviate the problems of the $p<0.05$ cutoff within the setting of hypothesis tests. Notably, \citet{benjamin2017} proposed to change the cutoff to $p<0.005$ in fields that have not already adopted a more stringent cutoff. However, several researchers have expressed concern that lowering the threshold would not fix the problems and may actually exacerbate them \citep{trafimow2018, lakens2018, amrhein2018remove}. Furthermore, while \citet{benjamin2017} describe the benefits of lowering the threshold, they also note that there may be other alternatives that do not involve hypothesis testing.

Another long-standing alternative is the Bayes factors \citep{jeffreys1935, jeffreys1961} (see \citet{kass1995} for an overview). The Bayes factor in favor of $H_0$ and against $H_1$ is $B_{01}(t) = \Pr(t | H_0) / \Pr(t | H_1)$, which can also be written as the ratio of the posterior odds in favor of $H_0$ to the prior odds in favor of $H_0$, i.e. $B_{01}(t) = [\Pr(H_0 | t) / \Pr(H_1 | t)] / [\Pr(H_0) / \Pr(H_1)]$. Conclusions based on the p-value do not always agree with conclusions based on the Bayes factor \citep{edwards1963, degroot1973, dickey1977, shafer1982}. For point null hypotheses, Bayes factors tend to be more conservative, i.e. Bayes factors provide less evidence against the null hypothesis than p-values \citep{berger1991}.

To conduct hypothesis tests with Bayes factors, one must use cutoff values to determine whether the observed data provides sufficient evidence to reject the null hypothesis. \citet[][Appendix B]{jeffreys1961} recommends cutoffs on the logarithmic scale for this purpose, and \citet{kass1995} note that the cutoffs proposed by \citet{jeffreys1961} are sensible in practice. Nonetheless, Bayes factors do require a cutoff threshold just as with p-values, which make Bayes factors prone to similar misuses.

\section{Exceedance probability for parameter estimates \label{motivation_framework}}

\subsection{Motivating use cases}

Suppose we are interested in estimating parameter $\bm{\theta} = (\theta_1, \ldots, \theta_d)^\tran$, particularly the $j^{\text{th}}$ element $\theta_j$, $1 \le j \le d$. Furthermore, suppose we are only interested in whether $\theta_j > c$ for some substantively meaningful cutoff $c$. Many scientific questions in psychology and other fields can be framed in this way. For example, $\theta_j$ might be the difference in response times of subjects to different stimuli, the difference in standardized test scores for students who undergo different curriculum, the increase in asthma rates per increase in ambient particulate matter, or the difference in tumor response rates between cancer patients who receive different treatments.

In all of the use cases above, if the effect size $\theta_j$ is greater than some substantively meaningful cutoff $c$, the result might warrant further study or action. Staying within a hypothesis testing framework, we could test the one-sided hypothesis $H_0: \theta_j \le c$ versus the alternative $H_1: \theta_j > c$. However, in many cases it is more informative to focus on estimation, particularly when $\theta_j$ is an interpretable quantity, because we believe this way of thinking tends to align with the mindset of many scientists.

Let $\hat{\theta}_j$ be an estimate of $\theta_j$ in an initial experiment or study. As a complement to hypothesis testing that focuses on estimation we could ask, ``given the results of the initial study, what is the probability of obtaining a $\hat{\theta}_j > c$ result in a replication study?" In a Bayesian framework, we could answer this question with the posterior predictive distribution \citep{bda2014, billheimer2019}.

In many cases a Bayesian approach may be useful for understanding uncertainty in future estimates, particularly to account for between-study heterogeneity. However, as described below, frequentist confidence intervals for the exceedance probability can also be informative, in that they describe the uncertainty in an estimate due to sampling variability alone.

\subsection{Framework \label{framework}}

Let $\Dn$ be a matrix of observed data from the initial study consisting of $n$ observations/rows. For example, in a regression problem we might have $\Dn = [\bm{y}, \bm{x}_1, \ldots, \bm{x}_d]$ where $\bm{y}, \bm{x}_j \in \mathbb{R}^n$ are the outcome and $j^{\text{th}}$ covariate, respectively. Let $\Dm$ be a matrix of data from a replication study with $m$ observations, i.e. a separate, independent dataset sampled from the same population as $\Dn$. For example, $\Dm = [\bm{y}^{\text{rep}}, \bm{x}^{\text{rep}}_1, \ldots, \bm{x}^{\text{rep}}_d]$ where $\bm{y}^{\text{rep}}, \bm{x}_j^{\text{rep}} \in \mathbb{R}^m$. Also, let $\bhthetan = \bm{\hat{\theta}}(\Dn)$ and $\bhthetarep = \bm{\hat{\theta}}(\Dm)$ be estimators of a parameter $\bm{\theta} = (\theta_1, \ldots, \theta_d)^\tran$ using datasets $\Dn$ and $\Dm$, respectively. We assume $\bhthetan$ and $\bhthetarep$ are estimated with the same procedure but different data.

We focus on normally distributed estimators with shared population parameters. Specifically, we assume that $\hthetanj \sim N(\theta_j, \sigma^2_j / n)$ and $\hthetarepj \sim N(\theta_j, \sigma^2_j / m)$, where both estimators have the same population parameters $\theta_j$ and $\sigma^2_j$, $1 \le j \le d$. The true exceedance probability that $\hthetarepj > c$ is
\begin{align}
\Pr_{\theta_j, \sigma_j}(\hthetarepj > c) &= \Pr \left( \sqrt{m} (\hthetarepj - \theta_j) / \sigma_j > \sqrt{m} (c - \theta_j) / \sigma_j \right) \nonumber \\
&= 1 - \Phi \left( \sqrt{m} (c - \theta_j) / \sigma_j \right).
\label{ep_true}
\end{align}

We aim to estimate (\ref{ep_true}) after collecting $\Dn$ but prior to collecting $\Dm$. Because we assume that $\hthetanj$ and $\hthetarepj$ share the same population parameters, we plug in $\hthetanj$ and $\hsigmanj$ to (\ref{ep_true}) to obtain the point estimate 
\begin{equation}
\Pr_{\hthetanj, \hsigmanj}(\hthetarepj > c) = 1 - \Phi \left( \sqrt{m} (c - \hthetanj) / \hsigmanj \right).
\label{ep_hat}
\end{equation}

The point estimate (\ref{ep_hat}) may not be reliable, particularly for small sample sizes or highly variable data, so it is crucial to consider confidence intervals that account for uncertainty both in $\hat{\theta}_j$ and $\hat{\sigma}_j$. For $\bhthetan$ and $\bhthetarep$ that are linear combinations of normal random variables, pointwise confidence intervals can be formed based on a t-distribution pivotal quantity as described in Section \ref{norm_ci}. In other cases, such as maximum likelihood estimators (MLEs), it may be possible to construct confidence intervals with normal approximations or bootstrap methods \citep{deMartini2013}. The confidence intervals can be reported for either a single scientifically meaningful cutoff $c$ or a series of $c$.

We allow for $m \ne n$, because even if the replication study aims to collect the same number of observations as the initial study, there might be some discrepancy due to a variety of data collection challenges or study design decisions. Consequently, it may be helpful to consider a few replication sample sizes $m$ near the initial study size $n$ to assess the sensitivity of results. 

We emphasize that $c$ is intended to represent a substantively meaningful value. While it may be challenging to select $c$ in many applied settings, doing so can better tie the statistical analysis to the scientific research question.

We also note that for a certain choice of hypothesis test and cutoff $c$, the exceedance probability (\ref{ep_true}) is equivalent to power. In particular, for the null $H_0: \theta_j \le \theta_{j0}$ and alternative $H_1: \theta_j > \theta_{j0}$, and setting $m=n$ and $c = \theta_{j0} + z_{1 - \alpha} \sigma_j / \sqrt{n}$ where $z_{1-\alpha} = \Phi(1 - \alpha)$, the exceedance probability (\ref{ep_true}) becomes $1 - \Phi(z_{1-\alpha} - \sqrt{n}(\theta_j - \theta_{j0}) / \sigma_j)$. This is the power of the test when $\sigma_j$ is known, also referred to as the population reproducibility probability \citep{boos2011}.

While this framework is motivated in terms of a replication study, the exceedance probability given by (\ref{ep_hat}) is purely a function of data collected in the initial study. Therefore, by setting $m=n$, confidence intervals for the exceedance probability can also be seen as expressing the uncertainty about $\theta_j$ on the probability scale with respect to a cutoff $c$, regardless of whether a replication study is feasible or planned.

\section{Linear combinations of normal random variables \label{linComboNorm}}

\subsection{Exceedance probability \label{finiteNormal}}

Suppose that $\bhthetan = \mat{A} \bm{y}$ for fixed $\mat{A} \in \mathbb{R}^{d \times n}$ and $\bm{y} \sim N(\bm{\mu}, \nu^2 \mat{V})$ where $\mat{V}$ is a known nonsingular and positive definite $n \times n$ matrix and $\text{E}[\bhthetan] = \bm{\theta}$, with an equivalent form for $\bhthetarep$. Then $\sqrt{n} (\bhthetan - \bm{\theta}) \sim N(\bm{0}, \mat{\Sigma})$ where $\mat{\Sigma} = n \nu^2 \mat{A} \mat{V} \mat{A}^\tran$ is the variance, with an analogous result for $\bhthetarep$. For example, for the sample mean of $n$ independent observations, we have $y_i \sim N(\mu, \nu^2)$, $i=1,\ldots,n$, $\mat{A} = (1/n,\ldots,1/n)$, and $\mat{\Sigma} = n \nu^2 \mat{A} \mat{A}^\tran = \nu^2$. For linear regression with design matrix $\mat{X} \in \mathbb{R}^{n \times d}$ and outcome $\bm{y} \sim N(\mat{X} \bm{\theta}, \nu^2 \mat{V})$, the general least squares estimate gives $\mat{A} = (\mat{X}^\tran \mat{V}^{-1} \mat{X})^{-1} \mat{X}^\tran \mat{V}^{-1}$ and $\mat{\Sigma} = n \nu^2 \mat{A} \mat{V} \mat{A}^\tran = n \nu^2 (\mat{X}^\tran \mat{V}^{-1} \mat{X})^{-1}$.

We estimate the marginal variance as $\hsigmasqnj = \hSigman_{jj}$ where $\hSigman = n \hnusq \mat{A} \mat{V} \mat{A}^\tran$ for $\hnusq = (n - d)^{-1} \|\bm{\hat{y}} - \bm{y} \|_2^2$ and fitted values $\bm{\hat{y}}$. Then as noted in Section \ref{framework}, we plug in $\hthetanj$ and $\hsigmasqnj$ to (\ref{ep_hat}) to obtain a point estimate for the marginal exceedance probability that $\hthetarepj > c$.

\subsection{Confidence intervals \label{norm_ci}}

Let $\tdist_{n-d, \delta}$ be the t-distribution with $n-d$ degrees of freedom and noncentrality parameter $\delta$. As shown in Appendix \ref{CI_sec}, which builds on \citet[][Appendix E.3.4]{meeker2017}, a two-sided $1-\alpha$ confidence interval for $\Pr_{\theta_j, \sigma_j}(\hthetarepj > c)$ is given by
\begin{equation}
\left[1-\Phi \left( \sqrt{\frac{m}{n}} \deltaU(c) \right), 1- \Phi \left( \sqrt{\frac{m}{n}} \deltaL(c) \right) \right],
\label{t_ci}
\end{equation}
where $\deltaL(c)$ and $\deltaU(c)$ are solutions to $\tdist_{n-d, \deltaL(c)}(q) = 1-\alpha/2$ and $\tdist_{n-d, \deltaU(c)}(q) = \alpha/2$ for
\begin{equation}
q = \sqrt{n}(c - \hthetanj) / \hsigmanj.
\label{q}
\end{equation}
Similarly, one-sided $1-\alpha$ confidence intervals can be obtained by solving $\tdist_{n-d, \deltaL(c)}(q) = 1-\alpha$ and $\tdist_{n-d, \deltaU(c)}(q) = \alpha$ to get bounds for upper and lower confidence intervals, respectively. As described in Appendix \ref{CI_sec}, these confidence intervals are formed from a t-distribution pivotal quantity and account for uncertainty in both $\hat{\theta}_j$ and $\hat{\sigma}_j$. Consequently, they are able to maintain their nominal coverage probability even in small samples when the model is correctly specified, as demonstrated in Appendix \ref{cov_prob_sec}.

\citet{meeker2017} focus on confidence intervals for the sample mean and $m = n$. However, as we show in Appendix \ref{CI_sec}, it is straightforward to extend the approach of \citet{meeker2017} to arbitrary linear combinations of normal random variables, $d>1$ mean parameters, and $m \ne n$.

We note that \citet{shao2002} and \citet{deMartini2008} provide confidence interval procedures for the reproducibility probability with a focus on test statistics, such as studentized means. \citet{shao2002} and \citet{deMartini2008} use essentially the same pivotal quantity as \citet{meeker2017} to obtain confidence intervals for the noncentrality parameter, but plug the lower bound for the noncentrality parameter into a t-distribution instead of a normal distribution to obtain a lower bound on power.

\subsection{Relationship to confidence intervals for $\theta$ \label{ci_relation}}

In this section, we analyze the relationship between confidence intervals for $\Pr_{ \theta_j, \sigma_j}(\hthetarepj > c)$ and confidence intervals for $\theta_j$. To simplify notation, throughout this section we drop the subscript $j$, though we assume that $\theta = \theta_j$ where $\bm{\theta} = (\theta_1, \ldots, \theta_d)^\tran$, $1 \le j \le d$. We also make the sample size in the replication study explicit in the notation as $\thetarepm$. We use $t_{n-d, 1-\alpha/2} = \tdist_{n-d, 0}^{-1}(1-\alpha/2)$ to denote the $1-\alpha/2$ quantile of the central t-distribution with $n-d$ degrees of freedom.

The results in this section given by Corollaries \ref{ep_half_corollary} and \ref{tube_corollary} are demonstrated in Figure \ref{ep_corollaries}. In Figure \ref{ep_corollaries}, the x-axis shows the cutoff $c$, the y-axis shows the exceedance probability, the solid black S-shaped curve shows the point estimate for the exceedance probability, and the gray area shows the point-wise confidence intervals for the exceedance probability. The pointwise confidence interval for a cutoff $c$ is given by the vertical slice through the plot that intersects the x-axis at $c$. For example, for $m=100$ and $c = \thetaL$, the point estimate is $\Pr_{\hat{\theta}_j, \hat{\sigma}_j}(\hthetarepj > \thetaL) \approx 1.0$ and the 95\% confidence interval for the exceedance probability is $[0.5, 1]$. The point estimate and 95\% confidence interval for $\hat{\theta}$ are also shown as a point and horizontal error bar, respectively.

We begin by stating Lemma \ref{delta_zero_lemma}, which is the basis for the subsequent results in this section. Throughout, we assume the data $\Dn$, sample size $n \in \mathbb{N}$ and estimates $\thetan \in \mathbb{R}$ and $\hsigman \in (0, \infty)$ are fixed.
\begin{lemma}
Let $\thetaL = \thetan - t_{n-d, 1-\alpha/2} \hsigman / \sqrt{n}$ and $\thetaU = \thetan + t_{n-d, 1-\alpha/2} \hsigman / \sqrt{n}$. Then $\deltaU(\thetaL) = 0$, and $\deltaL(\thetaU) = 0$.
\label{delta_zero_lemma}
\end{lemma}

\begin{proof}[Proof of Lemma \ref{delta_zero_lemma}]
Let $c = \thetaL$. Then the argument $q$ to the non-central t-distribution given by (\ref{q}) is
\begin{align*}
q &= \frac{\sqrt{n}(c - \thetan)}{\hsigman} \\
  &= \frac{\sqrt{n} \left( \thetan - t_{n-d, 1-\alpha/2} \hsigman / \sqrt{n} - \thetan \right)}{\hsigman} \\
  &= -t_{n-d, 1-\alpha/2}.
\end{align*}
Therefore, $\deltaU(\thetaL)$ is the solution to $\tdist_{n-d, \deltaU(\thetaL)}(-t_{n-d, 1-\alpha/2}) = \alpha/2$. By the symmetry of the central t-distribution about zero, we have $-t_{n-d, 1-\alpha/2} = t_{n-d, \alpha/2}$. Consequently, $\tdist_{n-d, \deltaU(\thetaL)}(-t_{n-d, 1-\alpha/2}) = \tdist_{n-d, \deltaU(\thetaL)}(t_{n-d, \alpha/2})$, and by definition $\tdist_{n-d, \deltaU(\thetaL)}(t_{n-d, \alpha/2}) = \alpha / 2$ if and only if $\deltaU(\thetaL) = 0$. This shows that $\deltaU(\thetaL) = 0$. An analogous argument shows that $\deltaL(\thetaU) = 0$, which proves the lemma.
\end{proof}

We now describe how confidence intervals for $\theta$ can be read from the plot of $\Pr_{\thetan, \hsigman}(\thetarepm > c)$ presented in Figure \ref{ep_corollaries}. From (\ref{ep_hat}), we have $\Pr_{\thetan, \hsigman}(\thetarepm > \thetan) = 0.5$ for all $m$. Corollary \ref{ep_half_corollary} gives a similar result for the confidence intervals around $\Pr_{\thetan, \hsigman}(\thetarepm > \thetaL)$ and $\Pr_{\thetan, \hsigman}(\thetarepm > \thetaU)$. We note that Corollary \ref{ep_half_corollary} in this paper is complementary to Corollaries 1 and 2 in \citet{deMartini2008}, though \citet{deMartini2008} focuses on hypothesis testing.

\begin{corollary}
Let $\thetaL$ and $\thetaU$ be as defined in Lemma \ref{delta_zero_lemma}. Then the lower bound of the two-sided $1-\alpha$ confidence interval around $\Pr_{\thetan, \hsigman}(\thetarepm > \thetaL)$ is equal to 0.5, and the upper bound of the two-sided $1-\alpha$ confidence interval around $\Pr_{\thetan, \hsigman}(\thetarepm > \thetaU)$ is equal to 0.5.
\label{ep_half_corollary}
\end{corollary}

\begin{proof}[Proof of Corollary \ref{ep_half_corollary}]
The two-sided $1-\alpha$ confidence interval about $\Pr_{\thetan, \hsigman}(\thetarepm > c)$ is given by $[1-\Phi(\sqrt{m/n} \deltaU(c)), 1-\Phi(\sqrt{m/n} \deltaL(c))]$ for $\deltaU$ and $\deltaL$ described in Section \ref{finiteNormal}. By Lemma \ref{delta_zero_lemma}, $\deltaU(\thetaL) = 0$ for all $m$. Therefore, for all $m$, the lower bound of the two-sided $1-\alpha$ confidence interval about $\Pr_{\thetan, \hsigman}(\thetarepm > \thetaL)$ is $1-\Phi(\sqrt{m/n} \deltaU(\thetaL)) = 1-\Phi(0) = 0.5$. An analogous argument shows that for all $m$, the upper bound of the two-sided $1-\alpha$ confidence interval about $\Pr_{\thetan, \hsigman}(\thetarepm > \thetaU)$ is equal to 0.5. This proves the corollary.
\end{proof}

As a consequence of Corollary \ref{ep_half_corollary}, and noting that $[\thetaL, \thetaU]$ as given in Lemma \ref{delta_zero_lemma} is a two-sided $1-\alpha$ confidence interval for $\theta$, it follows that the two-sided $1-\alpha$ confidence interval for $\theta$ can be read directly from the plot of $\Pr_{\thetan, \hsigman}(\thetarep > c)$ shown in Figure \ref{ep_corollaries}. This is done by drawing a horizontal line across the plot at $\Pr_{\thetan, \hsigman}(\thetarep > c) = 0.5$ on the y-axis and finding the leftmost and rightmost points $c$ at which the horizontal line intersects the confidence bands.

We now describe the asymptotic behavior of the confidence intervals for $\Pr_{\theta, \sigma}(\thetarepm > c)$ as $m$ goes to infinity. First, we note that as $m \rightarrow \infty$, $\Pr_{\thetan, \hsigman}(\thetarepm > c) \rightarrow 1$ for $c < \thetan$ and $\Pr_{\thetan, \hsigman}(\thetarepm > c) \rightarrow 0$ for $c > \thetan$. By Corollary \ref{tube_corollary}, the confidence interval around $\Pr_{\thetan, \hsigman}(\thetarepm > c)$ converges in a similar manner, which is demonstrated in Figure \ref{ep_corollaries}.

\begin{corollary}
Let $\thetaL$ and $\thetaU$ be as defined in Lemma \ref{delta_zero_lemma}. Then
\begin{equation}
1-\Phi \left(\sqrt{m/n} \deltaU(c) \right) \rightarrow
\begin{cases}
1 & c < \thetaL \\
0.5 & c = \thetaL \\
0 & c > \thetaL
\end{cases}
\quad \text{as} \quad m \rightarrow \infty
\label{cond1}
\end{equation}
and 
\begin{equation}
1-\Phi \left(\sqrt{m/n} \delta_L(c) \right) \rightarrow
\begin{cases}
1 & c < \thetaU \\
0.5 & c = \thetaU \\
0 & c > \thetaU
\end{cases}
\quad \text{as} \quad m \rightarrow \infty.
\label{cond2}
\end{equation}
\label{tube_corollary}
\end{corollary}

\begin{proof}[Proof of Corollary \ref{tube_corollary}]
By Lemma \ref{delta_zero_lemma}, $\deltaU(\thetaL) = 0$. Furthermore $\deltaU(c)$ is a strictly monotone increasing function of $c$. Consequently, $\deltaU(c) < 0$ for $c < \thetaL$, and $\deltaU(c) > 0$ for $c > \thetaL$. It follows that as $m \rightarrow \infty$, $\sqrt{m/n} \deltaU(c) \rightarrow -\infty$ for $c < \thetaL$ and $\sqrt{m/n} \deltaU(c) \rightarrow \infty$ for $c > \thetaL$. Therefore, as $m \rightarrow \infty$, $1-\Phi(\sqrt{m/n} \deltaU(c)) \rightarrow 1$ for $c < \thetaL$ and $1-\Phi(\sqrt{m/n} \deltaU(c)) \rightarrow 0$ for $c > \thetaL$. Furthermore, because $\deltaU(\thetaL) = 0$, we have $1-\Phi(\sqrt{m/n}\deltaU(\thetaL)) = 0.5$ for all $m$. This shows that the conditions in (\ref{cond1}) hold. An analogous argument shows that the conditions in (\ref{cond2}) hold, which proves the corollary.
\end{proof}

\begin{figure}[H]
\centering
\includegraphics[scale = 0.65]{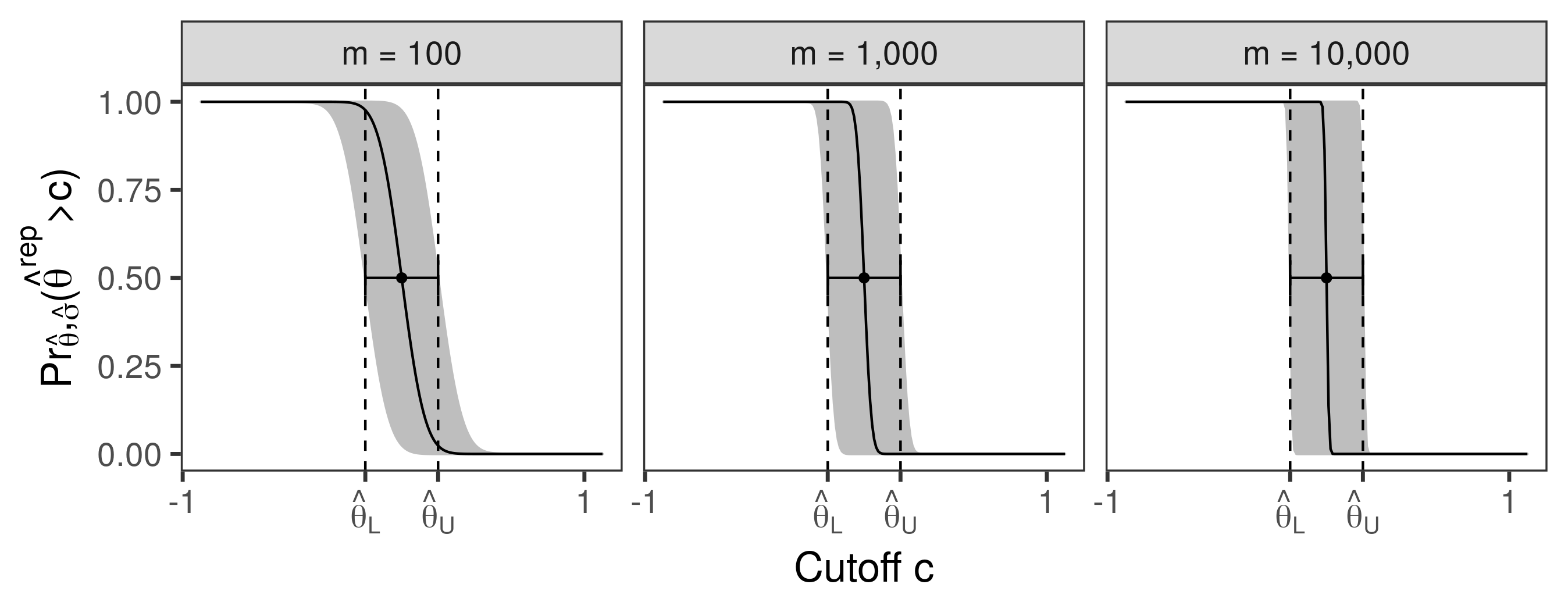}
\caption{Exceedance probability for the sample mean (data simulated as described in Section \ref{examples}) with $n = 100$. The solid black line shows $\Pr_{\thetan, \hsigman}(\thetarepm > c)$ and the gray area shows the 95\% pointwise confidence intervals. The pointwise confidence interval for a cutoff $c$ is given by the vertical slice through the plot that intersects the x-axis at $c$. The point estimate $\thetan$ and confidence interval $[\thetaL, \thetaU] = [\thetan \pm t_{n-1,1-\alpha/2} \hat{\sigma}_n/\sqrt{n}]$ for $\alpha = 0.05$ are shown by the single point and horizontal error bars at $\Pr_{\thetan, \hsigman}(\thetarepm > c) = 0.5$. Large $m$ shown to demonstrate Corollary \ref{tube_corollary}.}
\label{ep_corollaries}
\end{figure}

Corollary \ref{tube_corollary} provides a way to interpret the $1-\alpha$ confidence interval $[\thetaL, \thetaU]$ in terms of the estimation uncertainty in a replication study as the sample size of the replication study becomes large. In particular, as the sample size $m$ of the replication study goes to infinity, the probability of obtaining an estimate $\thetarepm \in [\thetaL, \thetaU]$ goes to $1-\alpha$. Conceptually, there is no sampling variability in the replication study in the limit as $m \rightarrow \infty$, so all sampling variability is from the initial study of size $n$. Because $[\thetaL, \thetaU]$ covers the true parameter $\theta$ with probability $1-\alpha$, it is not surprising that in the limit as $m \rightarrow \infty$, $[\thetaL, \thetaU]$ also covers $\thetarepm$ with probability $1-\alpha$.

This slightly different emphasis might be useful for teaching purposes to help reinforce the definition of confidence intervals. In particular, by emphasizing the uncertainty in a random but observable parameter estimate, as opposed to the uncertainty about a fixed but unobservable parameter value, this interpretation might be more accessible in application-oriented introductory settings. This interpretation requires that the replication study be identical to the initial study in all respects except for sample size.

While Corollary \ref{tube_corollary} demonstrates a connection between confidence intervals for the exceedance probability and confidence intervals for $\theta$ as $m \rightarrow \infty$, we emphasize that in practice we recommend choosing $m$ similar to $n$ to avoid misrepresenting the amount of certainty in the replication.

We also note that all of the above results can be adapted to one-sided confidence intervals. For example, lower confidence intervals for $\theta$ have the form $[\thetaL, \infty)$ where $\thetaL = \thetan - t_{n-d, 1-\alpha} \hsigman / \sqrt{n}$. Therefore, by replacing $t_{n-d, 1-\alpha/2}$ with $t_{n-d, 1-\alpha}$ in Lemma \ref{delta_zero_lemma}, the results in Corollary \ref{ep_half_corollary} and \ref{tube_corollary} related to $\thetaL$ also hold for one-sided $1-\alpha$ confidence intervals.

\subsection{Relationship to p-values}

In this section, we describe the relationship between confidence intervals for the exceedance probability and p-values. Similar to Section \ref{ci_relation}, throughout this section we drop the subscript $j$, though we assume that $\theta = \theta_j$ and $\sigma^2 = \Sigma_{jj}$ where $\bm{\theta} = (\theta_1, \ldots, \theta_d)^\tran$ and $\Sigma \in \mathbb{R}^{d \times d}$, $1 \le j \le d$. We note that \citet{deMartini2012} provides a complementary comparison with a focus on hypothesis testing.

\subsubsection{Two-sided hypothesis test \label{2_sided_p}}

For observed data $\bm{D}$, null hypothesis $H_0: \theta = c$, and alternative $H_1: \theta \ne c$, the p-value is given by $p_c(\bm{D}) = 2[1- \tdist_{n-d, 0}(\sqrt{n}|\hat{\theta} - c|/\hat{\sigma})]$. We treat $p_c(\bm{D})$ as a function of the cutoff $c$ for fixed data $\bm{D}$.

Figure \ref{ci_p_val_overlay} shows the p-value $p_c(\bm{D})$ (solid line) for the same data used in Figure \ref{ep_corollaries}. In this example, $\bm{D} = (y_1,\ldots,y_n)^\tran$ where $y_i \overset{\text{i.i.d}}{\sim} N(\theta, \sigma^2)$, $i=1,\ldots,n$ for $n = 100$, and we estimate $\hat{\theta} = \bar{y}$, $\hat{\sigma}^2 = (n-1)^{-1} \sum_{i=1}^n (y_i - \bar{y})^2$, and $[\thetaL, \thetaU] = \hat{\theta} \pm t_{n-1, 1-\alpha/2} \hat{\sigma}$. The quantity $1-p_c(\bm{D})$ is also called the confidence curve \citep[see][]{schweder2016}.

The pointwise 95\% confidence intervals for $\Pr_{\thetan, \hsigman}(\thetarep \le c)$ are also shown in Figure \ref{ci_p_val_overlay} (gray bands) for $m = n$. We show $\Pr_{\thetan, \hsigman}(\thetarep \le c)$ in Figure \ref{ci_p_val_overlay} instead of $\Pr_{\thetan, \hsigman}(\thetarep > c)$ to facilitate comparison with the p-value.

\begin{figure}[H]
\centering
\includegraphics[scale = 0.65]{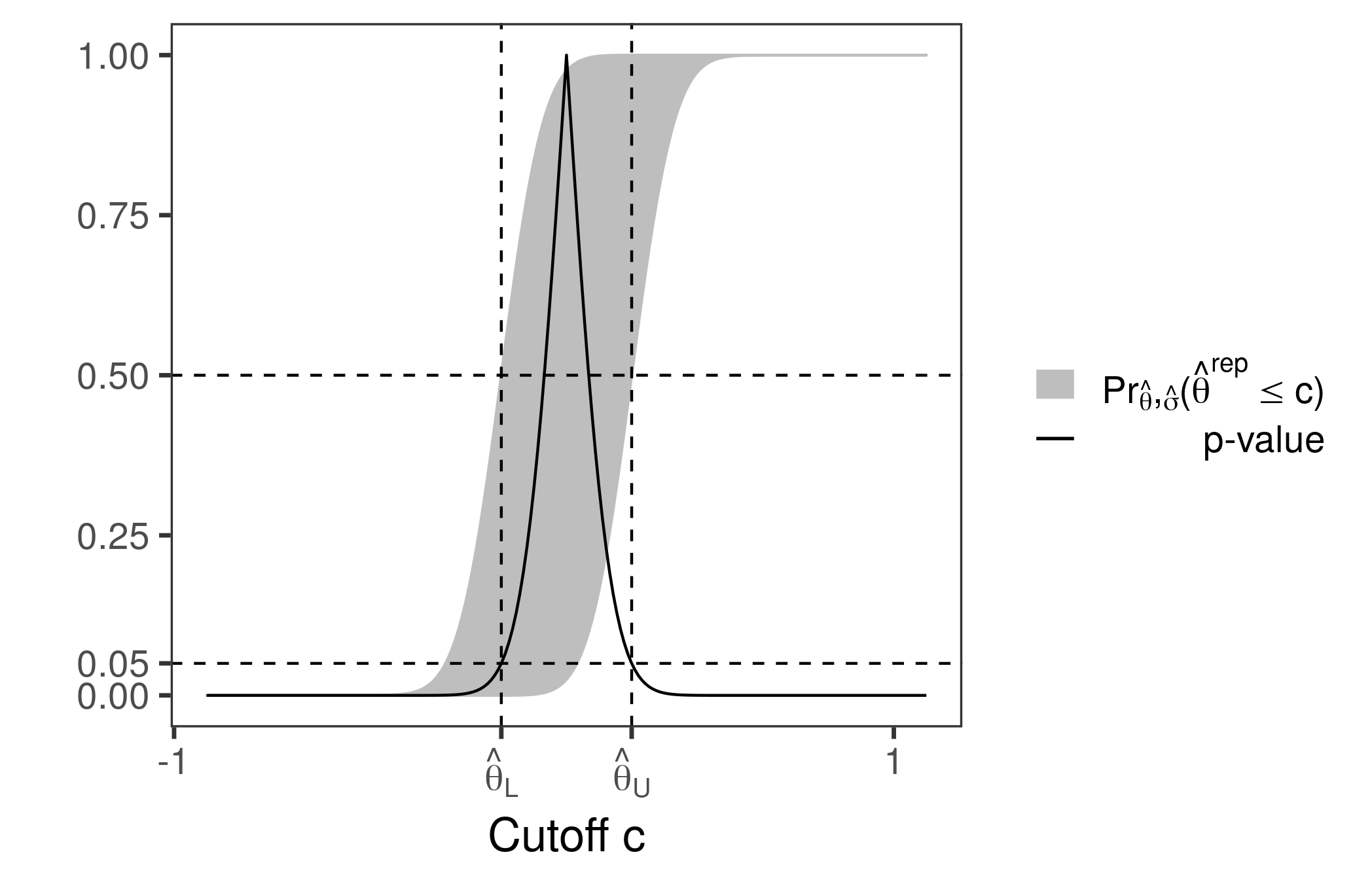}
\caption{$\Pr_{\thetan, \hsigman}(\thetarep \le c)$ compared to two-sided p-values for the sample mean with $n = m = 100$. The gray area shows the two-sided  95\% pointwise confidence intervals (the pointwise confidence interval for a cutoff $c$ is given by the vertical slice through the plot that intersects the x-axis at $c$). The solid line shows the two-sided p-value $p_c(\bm{D})$ for fixed data $\bm{D}$ and null hypotheses $H_0: \theta = c$, where $c$ is given on the x-axis.}
\label{ci_p_val_overlay}
\end{figure}

By definition, $\{\thetaL, \thetaU \} = \{c : p_c(\bm{D}) = \alpha \}$. In other words, the $1-\alpha$ two-sided confidence interval for $\theta$ can be found by drawing a horizontal line at $\alpha$ on the y-axis in Figure \ref{ci_p_val_overlay} and finding the leftmost and rightmost cutoffs $c$ at which the horizontal line intersects the solid curve. This is shown in Figure \ref{ci_p_val_overlay} for $\alpha = 0.05$.

From Corollary \ref{ep_half_corollary}, we have that for the two-sided $1-\alpha$ confidence interval around $\Pr_{\thetan, \hsigman}(\thetarep > c)$, the lower bound is 0.5 at $c = \thetaL$ and the upper bound is 0.5 at $c = \thetaU$. Equivalently, for the $1-\alpha$ confidence interval around $\Pr_{\thetan, \hsigman}(\thetarep \le c)$, the upper bound is 0.5 at $c = \thetaL$ and the lower bound is 0.5 at $c = \thetaU$. In other words, at the edges of the two-sided $1-\alpha$ confidence interval for $\theta$, the two-sided p-value is always equal to $\alpha$, and with $1-\alpha$ confidence the exceedance probability could be as low or high as 0.5. This is demonstrated in Figure \ref{ci_p_val_overlay} for $\alpha = 0.05$ by the intersections between the vertical and horizontal dashed lines.

The preceding discussion shows that when a p-value of $\alpha$ is obtained for the null hypothesis $H_0: \theta = c$ versus alternative $H_1: \theta \ne c$, then with $1-\alpha$ confidence there could be up to a 50\% probability of obtaining a future point estimate on the other side of $c$.

\subsubsection{One-sided hypothesis test}

For one-sided hypothesis tests, the p-value from a t-test is equivalent to the exceedance probability given by (\ref{ep_hat}) for large $n$. For example, suppose we wanted to test the null hypothesis $H_0: \theta \le c$ versus the alternative $H_1: \theta > c$. The p-value is given by $p_c(\bm{D}) = \tdist_{n-d, 0}(\sqrt{n}(c - \hat{\theta})/\hat{\sigma})$. For $m = n$, we have that $p_c(\bm{D}) \rightarrow \Pr_{\thetan, \hsigman}(\thetarep \le c)$ as $n \rightarrow \infty$. Notably, this makes it possible to use confidence intervals for the exceedance probability to obtain asymptotic pointwise confidence intervals for one-sided p-values.

For the null and alternative hypotheses above, the lower one-sided confidence interval is $[\thetaL, \infty)$ where $\thetaL = \hat{\theta} - t_{n-d, 1 - \alpha} \hat{\sigma}/\sqrt{n}$. As noted in Section \ref{ci_relation}, by replacing $t_{n-d, 1 - \alpha/2}$ with $t_{n-d, 1 - \alpha}$ in Lemma \ref{ep_half_corollary}, Corollary \ref{ep_half_corollary} shows that the upper bound of a one-sided $1-\alpha$ confidence interval for $ \Pr_{\thetan, \hsigman}(\thetarep \le \thetaL)$ is equal to 0.5. Because $p_c(\bm{D}) \approx \Pr_{\thetan, \hsigman}(\thetarep \le c)$ for large $n$, the upper one-sided $1-\alpha$ confidence interval for $p_{\thetaL}(\bm{D})$ is always bounded above by 0.5 for large $n$. In other words, if the one-sided p-value is equal to $\alpha$, the $1-\alpha$ confidence interval for the p-value is always $[0, 0.5]$. This is shown in Figure \ref{ci_p_val_overlay_one_sided} for $\alpha = 0.05$ using the same data as in Section \ref{2_sided_p}, and demonstrates the large amount of variability inherent to p-values near the 0.05 cutoff.

The results shown in Figure \ref{ci_p_val_overlay_one_sided} may seem at odds with known behaviors of p-values, such as being uniformly distributed under the null hypothesis and having an approximate lognormal distribution under alternative hypotheses if certain conditions apply, including asymptotic normality of the test statistic \citep{lambert1982, boos2011}. However, the viewpoint in Figure \ref{ci_p_val_overlay_one_sided} does not conflict with these well known results and in fact complements them, as described below.

\begin{figure}[H]
\centering
\includegraphics[scale = 0.65]{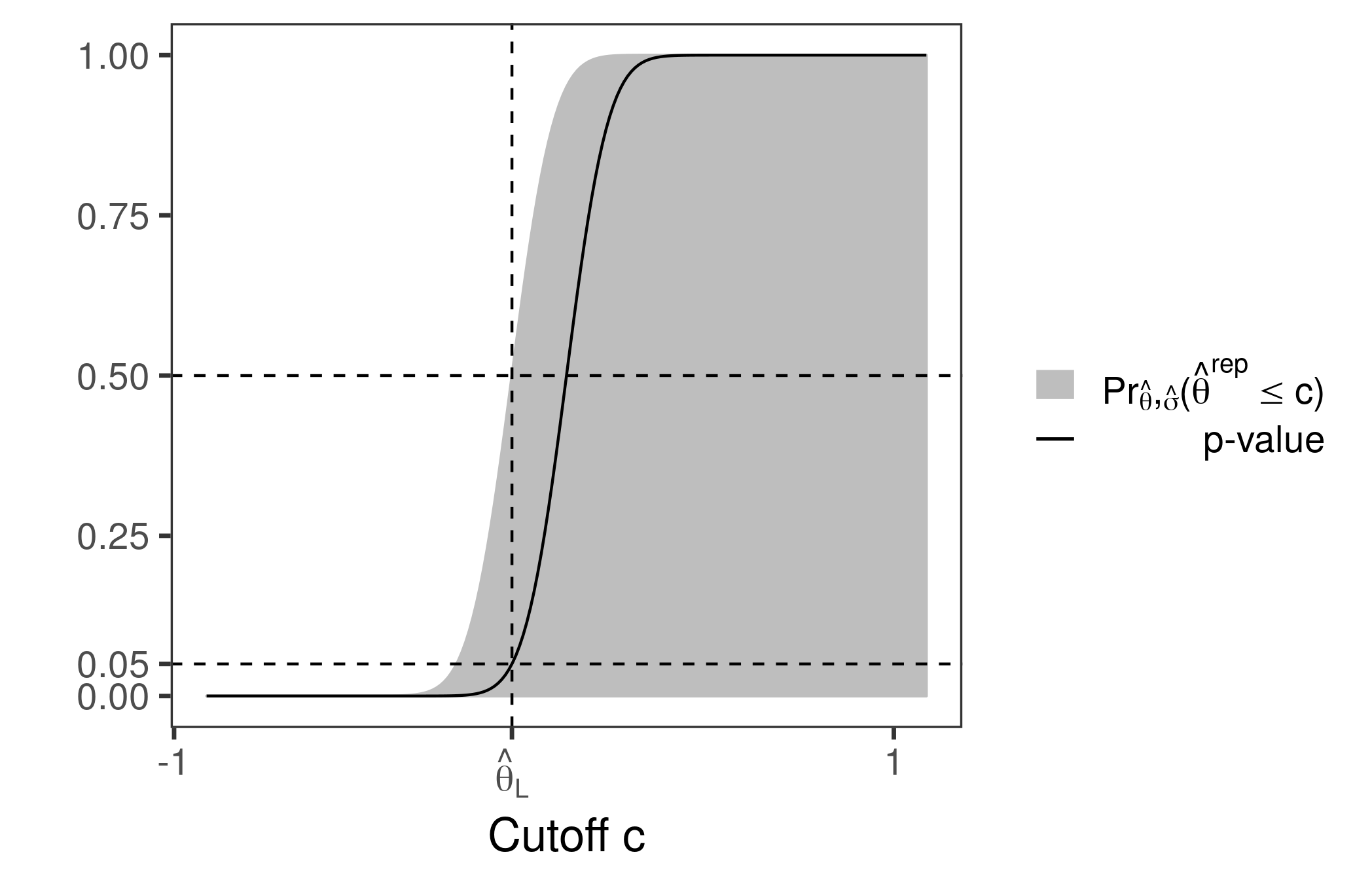}
\caption{$\Pr_{\thetan, \hsigman}(\thetarep \le c)$ and one-sided p-values for the sample mean with $n = m = 100$. The gray area shows the one-sided 95\% pointwise confidence intervals (the pointwise confidence interval for a cutoff $c$ is given by the vertical slice through the plot that intersects the x-axis at $c$). The solid curve shows the one-sided p-value $p_c(\bm{D})$ for fixed data $\bm{D}$ and null hypotheses $H_0: \theta \le c$ versus the alternative $H_1: \theta > c$, where $c$ is given on the x-axis.}
\label{ci_p_val_overlay_one_sided}
\end{figure}

Typically, the distribution of p-values is obtained upon repeated sampling of data $\bm{D}$ under a fixed null hypothesis (e.g. $H_0: \theta = 0$). In contrast, here we fix the data $\bm{D}$ and obtain confidence intervals for the p-value under different null hypotheses ($H_0: \theta = c$ for a series of $c$). 

Suppose in truth $\theta = \thetaL$ and we wanted to test the null $H_0: \theta \le \thetaL$ versus the alternative $H_1: \theta > \thetaL$. Upon sampling new data, the black solid curve in Figure \ref{ci_p_val_overlay_one_sided} would shift left or right and intersect the vertical dashed line at a different y-axis value, giving a new p-value. For example, if the shifted curve intersected the vertical dashed line at 0.2 on the y-axis, the p-value for that sample wold be 0.2. Over many such samples of data, we would obtain a uniform distribution of p-values between 0 and 1.

One can imagine that if in truth $\theta = 1$ instead of $\thetaL$ (where $\thetaL < 1$ as shown in Figure \ref{ci_p_val_overlay_one_sided}) then upon repeated sampling of data the black curve would tend to fall farther to the right then its current position and would intersect the vertical dashed line at much lower y-axis values. This would lead to a heavily left-skewed distribution of p-values, as expected under the alternative hypothesis.

\section{Examples \label{examples}}

\subsection{Simulated data \label{sim_example}}

In this section, we demonstrate how confidence intervals for the exceedance probability can be used in practice for the sample mean and how they compare to p-values, Bayes factors, and standard confidence intervals. For this example, we generated data $\Dn = (y_1, \ldots, y_n)^\tran$ where $y_i \overset{\text{i.i.d.}}{\sim} N(\theta, \sigma^2)$, $i=1,\ldots,n$, for $\theta = 0$ and $\sigma^2 = 1$. We then set $\thetan = \bar{y}$ and $\hsigmasqn = (n-1)^{-1} \sum_{i=1}^n(y_i - \bar{y})^2$.

Figure \ref{sample_mean_100} shows the simulated data for $n = 100$ observations ($\bar{y} = 0.25, \text{sd} = 1.1$) and Figure \ref{ep_mean_100} shows the exceedance probabilities with pointwise 95\% confidence intervals. In Figure \ref{ep_mean_100}, the x-axis shows the cutoff value $c$, the y-axis shows the exceedance probability, the solid black S-shaped curve shows the point estimate of the exceedance probability, and the gray area shows the 95\% pointwise confidence intervals. The pointwise confidence interval for a cutoff $c$ is given by the vertical slice through the plot that intersects the x-axis at $c$. For example, at $c=0$ the point estimate is $\Pr_{\hat{\theta}_j, \hat{\sigma}_j}(\hthetarepj > 0) \approx 1.0$ and the 95\% confidence interval is $[0.58, 1]$.

\begin{figure}[H]
\centering
\includegraphics[scale = 0.6]{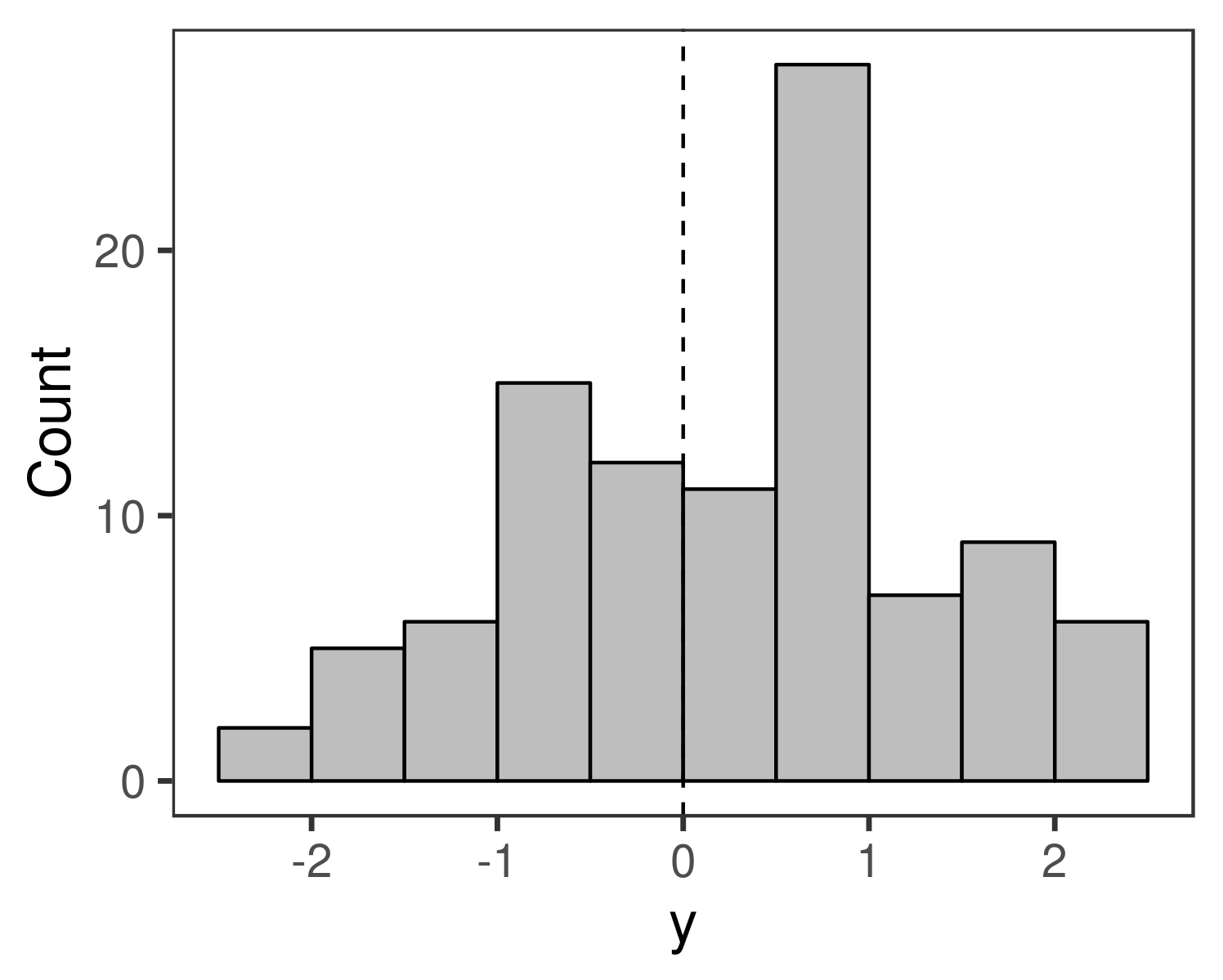}
\caption{Histogram of generated data and true mean (dashed line), $n = 100$.}
\label{sample_mean_100}
\end{figure}

\begin{figure}[H]
\centering
\includegraphics[scale = 0.65]{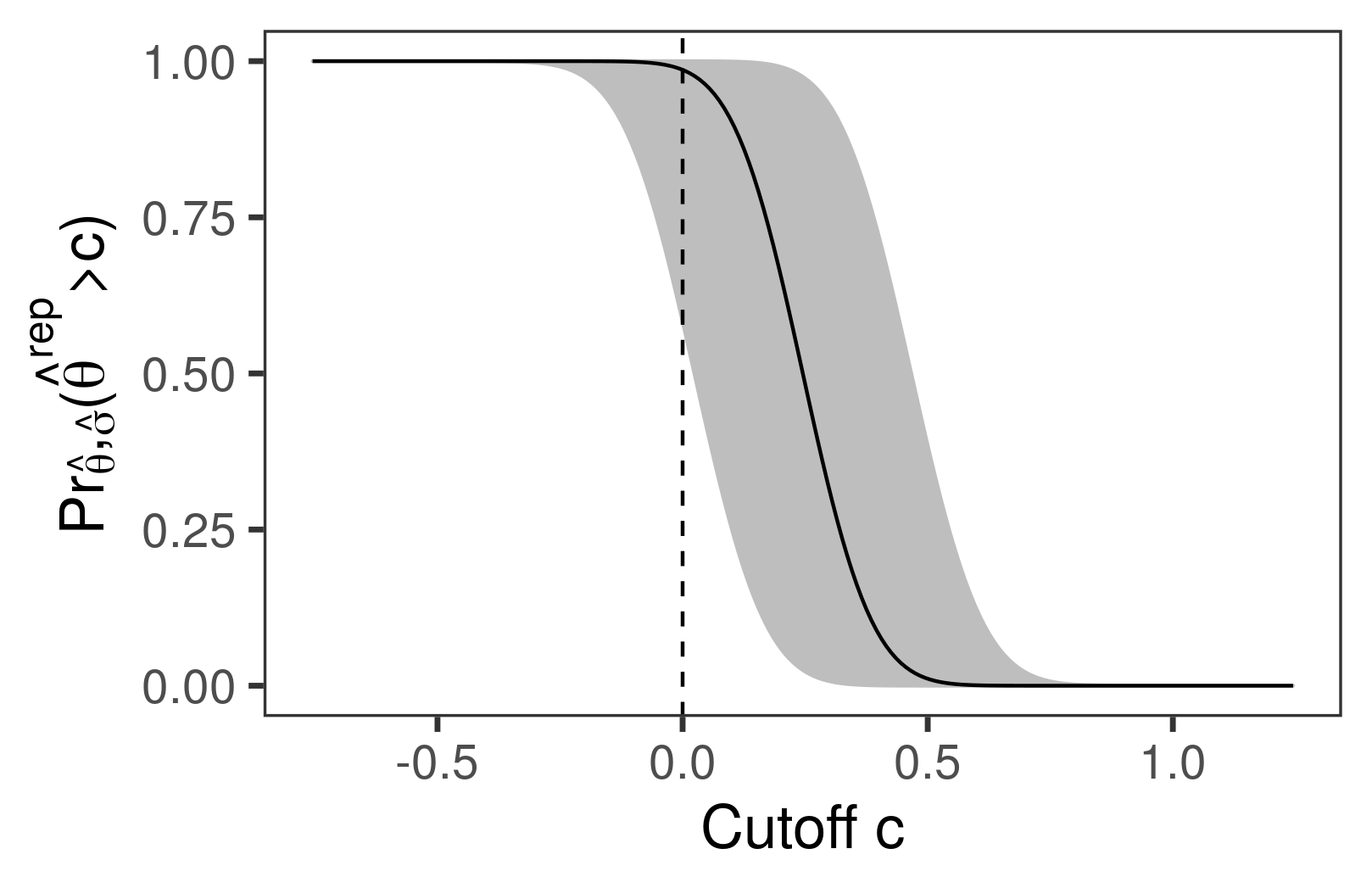}
\caption{Exceedance probability for $m = n = 100$ with pointwise 95\% confidence intervals. The pointwise confidence interval for a cutoff $c$ is given by the vertical slice through the plot that intersects the x-axis at $c$.}
\label{ep_mean_100}
\end{figure}

Suppose we wanted to test the null hypothesis $H_0: \theta \le 0$ versus the alternative $H_1: \theta > 0$. A one-sided t-test gives a p-value of 0.015 and the 1-sided 95\% confidence interval is $[0.06, \infty)$, so we would incorrectly reject $H_0$ under the standard 0.05 significance level. We also computed a Bayes factor for the one-sided null hypothesis using the \verb|BayesFactor| package \citep{BayesFactor} with a Cauchy prior on the standardized effect size and a non-informative Jeffreys prior on the variance, also called a JZS prior \citep{rouder2009, morey2011}. As  \citet{rouder2009} note, the ``JZS prior is designed to minimize assumptions about the range of effect size, and in this sense it is an objective prior.'' Using the JZS prior, we obtained a Bayes factor in favor of $H_0$ of $B_{01} = 0.016$. According to \citet{kass1995}, this is strong evidence against the null hypothesis ($1/B_{01} = 60.7$). 

However, from Figure \ref{ep_mean_100} we see that with 95\% certainty, $\Pr_{\thetan, \hsigman}(\thetarep > 0)$ could be as low as 58\% for a replication study with $m =100$ observations. In this example, the p-value and Bayes factor provide confidence that $\theta > 0$, but there is a reasonable chance that a future point estimate of $\theta$ will be less than 0. In this situation, considering the exceedance probability together with its confidence interval could help researchers to avoid making claims that may be refuted by future experiments.

We note that whereas the exceedance probability is itself directional (being the probability of observing an estimate larger than a cutoff), the confidence intervals depicted by the gray bands in Figure \ref{ep_mean_100} are two-sided and not directional. This is a key distinction to keep in mind when interpreting Figure \ref{ep_mean_100}.

It is most natural to compare the exceedance probability against one-sided confidence intervals, p-values, and Bayes factors, because they are all directional. However, for completeness we can also compare with two-sided metrics. In this example, the two-sided 95\% confidence interval for $\theta$ is $[0.024, 0.47]$ and the two-sided p-value is $0.03$. This shows that under a 0.05 significance level we would also reject the point null hypothesis $H_0: \theta = 0$ in favor of the alternative $H_1: \theta \ne 0$. In contrast, a two-sided Bayes factor computed with the \verb|BayesFactor| package \citep{BayesFactor} and JZS prior gives a result of $B_{01} = 0.91$, which provides no evidence for either the null or alternative hypothesis.

Of the p-values, standard confidence intervals, and Bayes factors above, all but the Bayes factor for the two-sided null result in an incorrect rejection of the null. Confidence intervals for the exceedance probability add context by showing that in a replication study there is a reasonable chance of obtaining a point estimate that is of the opposite sign based on sampling variability alone. This result could be distilled into the message, ``while the confidence interval only contains positive values, there is a decent chance that in a replication study the point estimate would be negative.'' In this way, confidence intervals for the exceedance probability can complement other inferential procedures.

The data in this example are unusual in that the sample mean is 2.27 standard deviations above the true population mean of 0. However, that is also the point. Due to publication bias and related issues, there are many more p-values just under 0.05 reported in psychological studies than would be expected by chance alone \citep{kuhberger2014}. In other words, many of the results reported in scientific journals are likely based on unusual data sets, in the same way that this example is unusual. We think this example helps to shed light on this issue, and that considering confidence intervals for the exceedance probability together with other metrics might help to mitigate this problem.

We also note that these different metrics (p-values, standard confidence intervals, Bayes factors, and confidence intervals for the exceedance probability) answer different questions. While they are related, it is not necessarily unexpected that they give different conclusions. By considering these metrics together, researchers can draw more informed scientific conclusions.

\subsection{Data from the Open Science Collaboration \label{application}}

We analyzed a dataset that was collected as part of the Open Science Collaboration effort to replicate 100 psychological studies \citep{open2015}. In particular, we re-analyzed data from the replication study of \citet{meixner2015} (available at \url{https://osf.io/atgp5}), which aimed to replicate the results of \citet{berry2008}. We did not have access to data from the original study of \citet{berry2008}, so we only analyzed data from the replication study of \citet{meixner2015}.

In the replication study of \citet{meixner2015}, each of 32 female volunteers was primed with an initial list of 70 words (the first and last ten of which were not used in later trials). The words were shown sequentially in a specific manner described by \citet{meixner2015}. Each volunteer was then shown a second list of 100 words, one word after the other, and asked to identify as quickly as possible whether the word in the new list was also in the initial list (50 were in the initial list and 50 were not). To replicate the main findings of \citet{berry2008}, \citet{meixner2015} hypothesized that the mean response time (RT) of misses (incorrectly saying a word was not in the initial list) would be quicker than the mean RT of correct rejections (correctly saying a word was not in the initial list).

Data from the replication study, accessible via the \textsf{R} script at \url{https://osf.io/9ivaj}, has one row per study volunteer. The relevant columns are \verb|meanRT_miss| and \verb|meanRT_cr|, the mean RT for misses and correct rejections, respectively, in milliseconds (ms). \citet{meixner2015} proceed by conducting a paired t-test, or equivalently, a t-test for whether $\bar{d}_i = \bar{y}_{i1} - \bar{y}_{i2}$ have mean zero, where $\bar{y}_{i1}$ and $\bar{y}_{i2}$ are the mean RT for volunteer $i$ for correct rejections and misses, respectively, and each $\bar{d}_i$ is treated as a single observation. Analyzing the data in this way may under-represent the number of observations, though this is not likely to alter the conclusions of \citet{meixner2015} regarding statistical significance.

Data from \citet{meixner2015} are shown in Figure \ref{application_data_hist}. A t-test for the two-sided null hypothesis $H_0: \theta = 0$ versus the alternative $H_1: \theta \ne 0$ where $\theta$ is the population difference in RTs gives a p-value of 0.023 (95\% confidence interval of $[8.65, 107]$), which confirms the finding of \citet{berry2008}. We also computed a Wilcoxon signed rank test of the null of symmetry about 0 and obtained the nearly identical p-value of 0.022, and used the \verb|BayesFactor| package \citep{BayesFactor} with the JZS prior to compute a two-sided Bayes factor in favor of the null of $B_{01} = 0.45$, which provides no evidence for either the null or alternative.

\begin{figure}[H]
\centering
\includegraphics[scale = 0.6]{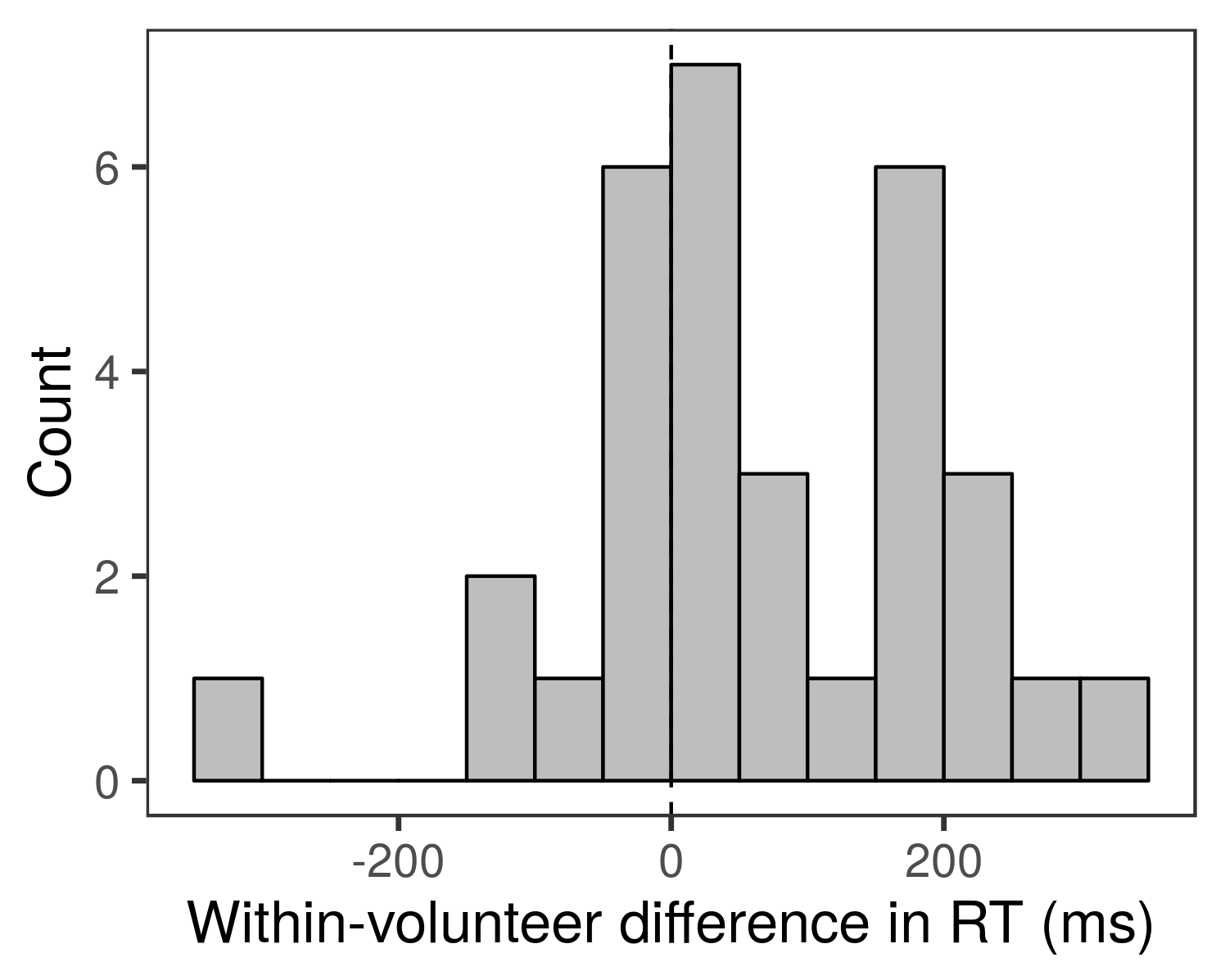}
\caption{Histogram of within-volunteer difference in RT (ms)}
\label{application_data_hist}
\end{figure}

Figure \ref{application_exceedance} shows the exceedance probability for $m = n = 32$ volunteers, along with two-sided 95\% pointwise confidence intervals. As shown in Figure \ref{application_exceedance}, with 95\% confidence the probability of obtaining a point estimate greater than zero in a future replication study could be as low as 63\%. In other words, with 95\% confidence there could be as high as a 37\% chance of a sign change in a future replication based on sampling variability alone. This conclusion is based only on sampling variability, and does not account for other aspects of the future replication study that may differ from the current one.

\begin{figure}[H]
\centering
\includegraphics[scale = 0.65]{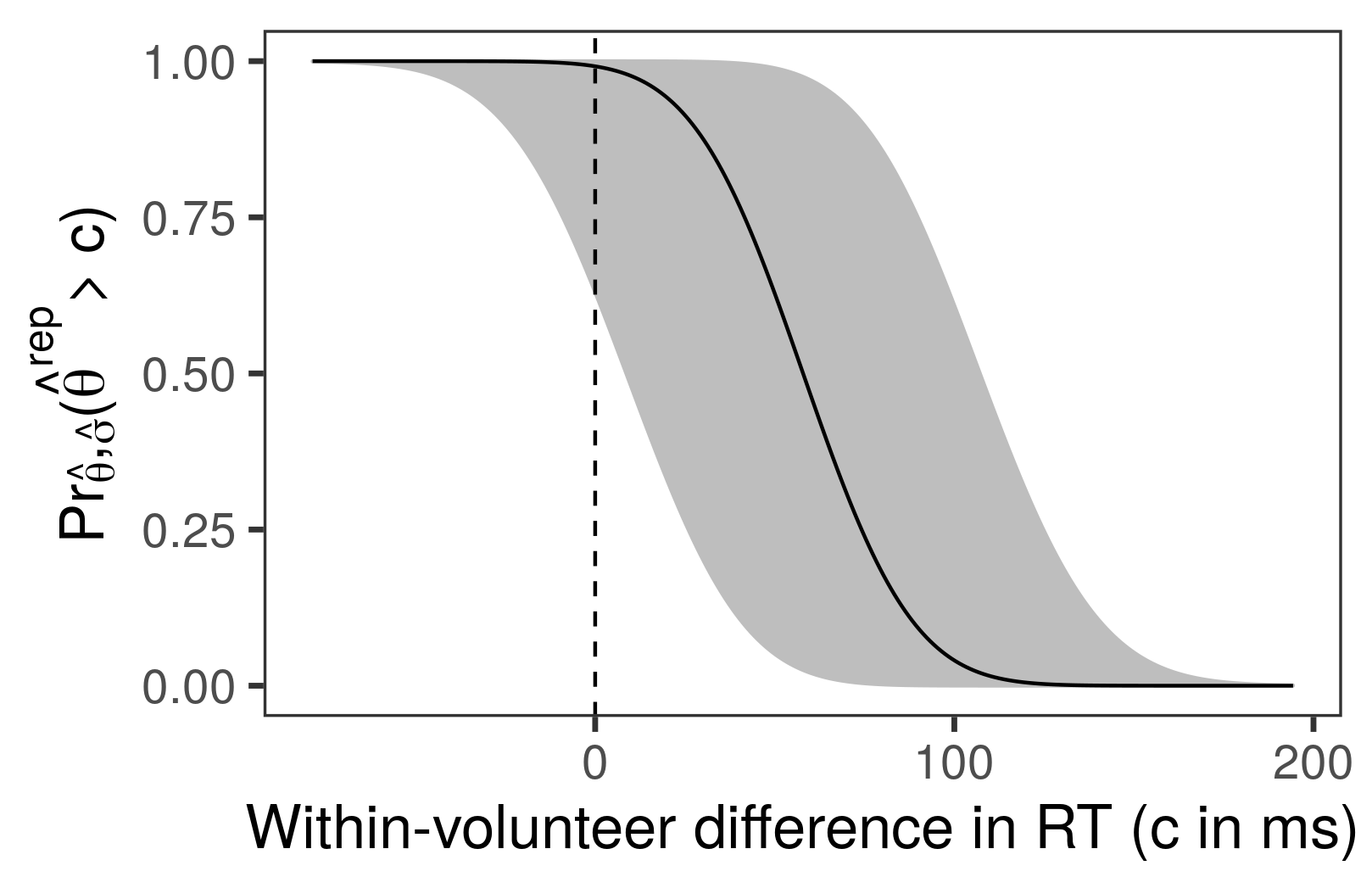}
\caption{Exceedance probability for $m=n=32$ volunteers with pointwise 95\% confidence intervals. The pointwise confidence interval for a cutoff $c$ is given by the vertical slice through the plot that intersects the x-axis at $c$.}
\label{application_exceedance}
\end{figure}

We also calculated one-sided confidence intervals, p-values, and Bayes factors for the null hypothesis $H_0: \theta \le 0$ versus the alternative $H_1: \theta > 0$. This gave a p-value of 0.011, a 95\% confidence interval of $[17.0, \infty)$ and a  Bayes factor in favor of the null of $B_{01} = 0.014$. In all cases we would reject the null hypothesis.

Of the p-values, confidence intervals, and Bayes factors above, all but the Bayes factor for the two-sided null hypothesis reject the null. However, as shown above, there is a reasonable chance that based on sampling variability alone, a future replication study would result in an estimate that is of the opposite sign.

\section{Conclusions \label{conclusion}}

In many situations, confidence intervals for the exceedance probability provide a simple, interpretable, and scientifically relevant metric that is geared towards replication. In particular, the exceedance probability is the same as power for a certain choice of one-sided hypothesis and cutoff, which is a familiar concept. Furthermore, when cutoffs are chosen to be substantively meaningful, confidence intervals for the exceedance probability provide likely probabilities of observing a substantively meaningful result in an exact replication study. This may help researchers to understand the stability of their result and complements standard confidence intervals and p-values.

The exceedance probability can only be interpreted as the reproducibility probability if the replication study is exactly the same as the initial study and enrolls participants from the exact same population. This may not hold in practice for a field such as psychology, in which there can be large between-study heterogeneity, which could cause the confidence intervals given by (\ref{t_ci}) to underestimate the amount of uncertainty in a future replication. However, as demonstrated in this article, confidence intervals for the exceedance probability can still provide an important perspective on estimation uncertainty that is not conveyed by standard confidence intervals and p-values.

We are not advocating that confidence intervals for the exceedance probability replace other metrics such as p-values and standard confidence intervals. Rather, we think that when used in conjunction with other methods, confidence intervals for the exceedance probability can provide a more complete understanding of estimation uncertainty.

We also note that confidence intervals for the exceedance probability can be used for hypothesis testing, just as standard confidence intervals can be used for the same purpose \citep{deMartini2008, deMartini2012, deMartini2013}. However, we think that in many applications it is more informative, relevant, and scientifically sound to use confidence intervals to describe the exceedance probabilities most compatible with the data as opposed to as a tool for making a yes/no decision. As \citet{amrhein2019retire} note, ``values just outside the interval do not differ substantively from those just inside the interval'' and ``not all values inside are equally compatible with the data, given the assumptions.'' The same caution is warranted here.

On a pedagogical note, the asymptotic behavior of confidence intervals for the exceedance probability as the size of an exact replication study becomes large might be useful for teaching purposes. In particular, this might help to reinforce the concept of confidence intervals in application-oriented introductory settings by emphasizing the uncertainty in a random but observable parameter estimate, as opposed to the uncertainty about a fixed but unobservable parameter value. However, in applications we recommend selecting $m$ similar in size to $n$ to accurately portray uncertainty.

In future work, it would be helpful to evaluate the performance of different confidence interval procedures for asymptotically normal estimators, such as MLEs. For example, using normal instead of t-distribution pivotal quantities, or bootstrap methods. We think it would also be helpful to investigate procedures for forming simultaneous confidence intervals for multi-parameter exceedance probabilities, and to compare the intervals given by (\ref{t_ci}) against Bayesian predictive methods. In the context of clinical trials, \citet{shao2002} found that frequentist confidence intervals for the exceedance probability were more conservative than Bayesian predictive methods, and it could be useful to study these trade-offs in the setting of psychological research as well.

\section{Supplementary material}

The R package \texttt{exceedProb} implements confidence intervals for the exceedance probability and is available on the CRAN. All code for reproducing examples and simulations in this paper is available at \url{https://github.com/bdsegal/code-for-exceedance-paper}.

\section{Acknowledgements}

I would like to thank Michael R. Elliott and Peter D. Hoff for their extremely helpful feedback and suggestions, as well as the editor, associate editor, and reviewers for their constructive comments which greatly improved the paper.

\begin{appendices}

\section{Derivation of confidence intervals \label{CI_sec}}

This appendix follows \citet[][Section E.3.4]{meeker2017} with the addition that we introduce the scaling factor $\sqrt{m/n}$ to allow for $m \ne n$ and show that the result holds for any linear combination of normal random variables and $d > 1$ mean parameters. Suppose $\bhthetan, \bhthetarep \in \mathbb{R}^d$ are linear combinations of normal random variables as in Section \ref{linComboNorm}. In particular, $\bhthetan = \mat{A} \bm{y}$ for fixed $\mat{A} \in \mathbb{R}^{d \times n}$ and  $\bm{y} \sim N(\bm{\mu}, \nu^2 \mat{V})$ where $\mat{V}$ is a known nonsingular and positive definite $n \times n$ matrix and $\text{E}[\bhthetan] = \bm{\theta}$, with an equivalent assumption for $\bhthetarep$. As described in Section \ref{finiteNormal}, the marginal exceedance probability for $\hthetarepj$, $1 \le j \le d$, is $\Pr_{\theta_j, \sigma_j}(\hthetarepj > c) = 1 - \Phi \left(\sqrt{m}(c - \theta_j)/\sigma_j \right)$ where $\Phi$ is the standard normal CDF and $\sigma^2_j = n \nu^2 (\mat{A} \mat{V} \mat{A}^\tran)_{jj}$. Let 
\[
Z = \frac{\sqrt{n}(\theta_j - \hthetanj)}{\sigma_j}, \quad \delta(c) = \frac{\sqrt{n}(c - \theta_j)}{\sigma_j}, \quad \text{and} \quad S = \frac{(n-d) \hsigmasqnj}{\sigma^2_j}
\]
where $\hsigmasqnj = n \hnusq (\mat{A} \mat{V} \mat{A}^\tran)_{jj}$, $\hnusq = (n - d)^{-1} \| \bm{\hat{y}} - \bm{y}\|^2_2$, and $\bm{\hat{y}}$ are the fitted values. Also let
\[
Q = \frac{\sqrt{n} (c - \hthetanj)}{\hsigmanj} = \frac{Z + \delta(c)}{\sqrt{S / (n-d)}}.
\]
We note that 
\[
\frac{\hsigmasqnj}{\sigma^2_j} = \frac{n \hnusq (\mat{A} \mat{V} \mat{A}^\tran)_{jj}}{n \nu^2 (\mat{A} \mat{V} \mat{A}^\tran)_{jj}} = \frac{\hnusq}{\nu^2}.
\]
Therefore, 
\[
S = \frac{(n-d) \hsigmasqnj}{\sigma^2_j} = \frac{(n-d) \hnusq}{\nu^2} \sim \chi^2_{n-d}
\]
where $\chi^2_{n-d}$ is the chi-squared distribution with $n-d$ degrees of freedom. We also have $Z \sim N(0,1)$ and $Z \perp S$. It follows that $Q \sim \tdist_{n-d, \delta(c)}$.

We note that $\tdist_{n-d, \delta(c)}$ is strictly monotone decreasing in $\delta(c)$, and serves as a pivotal quantity \citep[][Theorem 7.1]{shao2003}. Consequently, a two-sided $1-\alpha$ confidence interval for $\delta(c)$ is given by $[\deltaL(c), \deltaU(c)]$ where $\tdist_{n-d, \deltaL(c)}(q) = 1 - \alpha / 2$ and $\tdist_{n-d, \deltaU(c)}(q) = \alpha / 2$ for observed value $q = \sqrt{n} (c - \hthetanj) / \hsigmanj$. We also note that $\Pr_{\theta_j, \sigma_j}(\hthetarepj > c) = 1- \Phi (\sqrt{m/n} \delta(c))$ is strictly monotone decreasing in $\delta(c)$ for fixed $m$ and $n$. Consequently, a two-sided $1-\alpha$ confidence interval for $\Pr_{\theta_j, \sigma_j}(\hthetarepj > c) = 1-\Phi(\sqrt{m/n} \delta(c))$ is given by $[1-\Phi(\sqrt{m/n} \deltaU(c)), 1- \Phi(\sqrt{m/n} \deltaL(c))]$. This is the result shown in (\ref{t_ci}) of Section \ref{finiteNormal}.

\section{Coverage probability simulations \label{cov_prob_sec}}

In this appendix, we investigated the coverage probability of intervals given by (\ref{t_ci}) for the sample mean and linear regression. For $k=1,\ldots,K$, we generated data $\mat{D}^k$ and estimated $\bhthetan^k$ and $\hSigman^k$ with data $\mat{D}^k$. We then estimated the coverage probability at cutoff $c$ as $\hat{P}(c) = K^{-1} \sum_{k=1}^K \mathbbm{1} \left[ \Pr_{\theta_j, \sigma_j}(\hthetarepj > c) \in I^k_c \right]$ for intervals $I^k_c$ formed with (\ref{t_ci}). Throughout, we set $\alpha = 0.05$ for a nominal coverage probability of 0.95.

\subsection{Sample mean \label{cp_mean}}

We generated data in the same manner as in Section \ref{examples}. In particular, for each of $k=1,\ldots,K$, we generated data $\mat{D}^k = (y^k_1, \ldots, y^k_n)^\tran$ where $y^k_i \overset{\text{i.i.d.}}{\sim} N(\theta, \sigma^2)$, $i=1,\ldots ,n$, for $\theta = 0$ and $\sigma^2 = 1$. Consequently, the true exceedance probability is $\Pr_{\theta, \sigma}(\thetarep > c) = 1 - \Phi(\sqrt{m}c)$.

Results from $K = 10,000$ simulated datasets with $n = 20, 40, 60, 80, 100$ observations and a replication sample size of $m = n$ are shown in Table \ref{mean_table}. As seen in Table \ref{mean_table}, the confidence intervals achieve their nominal coverage probability.

\begin{table}[H]
\centering
\caption{Empirical coverage probability for $1 - \Phi(\sqrt{m}c)$ for different samples sizes, $m = n$, and cutoffs $c$ for the sample mean. $K = 10,000$ datasets generated for each sample size $n$.}
\begin{tabular}{cccccc}
\hline \hline
 & \multicolumn{5}{c}{$n$} \\
 \cline{2-6}
\hline
$c$ & 20    & 40    & 60    & 80  & 100 \\
-0.5 & 0.953 & 0.945 & 0.951 & 0.946 & 0.951 \\
-0.4 & 0.953 & 0.946 & 0.950 & 0.946 & 0.950 \\
-0.3 & 0.953 & 0.946 & 0.950 & 0.946 & 0.951 \\
-0.2 & 0.953 & 0.947 & 0.949 & 0.946 & 0.951 \\
-0.1 & 0.952 & 0.948 & 0.949 & 0.945 & 0.951 \\
 0.0 & 0.952 & 0.948 & 0.949 & 0.946 & 0.951 \\
 0.1 & 0.952 & 0.948 & 0.948 & 0.945 & 0.951 \\
 0.2 & 0.951 & 0.946 & 0.949 & 0.947 & 0.952 \\
 0.3 & 0.951 & 0.947 & 0.949 & 0.946 & 0.951 \\
 0.4 & 0.951 & 0.947 & 0.949 & 0.947 & 0.951 \\
 0.5 & 0.950 & 0.947 & 0.950 & 0.947 & 0.951
\end{tabular}
\label{mean_table}
\end{table}
\subsection{Linear regression \label{cp_slr}}

We set the design matrix to $\mat{X} = [\bm{1}, \bm{x}]$ for $n \times 1$ vectors $\bm{1} = (1, \ldots, 1)^\tran$ and $\bm{x} = (x_1, \ldots, x_n)^\tran$ where $\bm{x}$ was fixed for all simulations of the same sample size $n$ ($x_i$ were initially generated as i.i.d. uniform(0, 10) random variables). We set the regression coefficients to $\bm{\theta} = (\theta_1, \theta_2)^\tran = (1, 2)^\tran$.  For $k=1, \ldots, K$, we generated responses as $\bm{y}^k \sim N(\mat{X} \bm{\theta}, \nu^2 \mat{I}_n)$ for variance $\nu^2 = 25$. We then fit a linear model to obtain $\bhthetan^k = (\mat{X}^\tran \mat{X})^{-1} \mat{X}^\tran \bm{y}^k$ and estimated the variance as $\hSigman^k = n \hat{\nu}^{2,k} (\mat{X}^\tran \mat{X})^{-1}$ where $\hat{\nu}^{2,k} = (n-2)^{-1} \| \bm{y}^k - \bm{\hat{y}}^k \|_2^2$ and $\bm{\hat{y}}^k = \mat{X} \bhthetan^k$.

In truth, we have $\hat{\theta}^k_2 \sim N(\theta_2, \sigma^{2}_2 / n)$ where $\sigma^{2}_2 = n \nu^2 (\mat{X}^\tran \mat{X})^{-1}_{2,2}$. Consequently, the true exceedance probability is $\Pr_{\theta_2, \sigma_2}(\hat{\theta}^{\text{rep}}_2 > c) = 1 - \Phi(\sqrt{m}(c - 2) / \sqrt{n 25 (\mat{X}^\tran \mat{X})^{-1}_{2,2}})$.

Results from $K = 10,000$ simulated datasets with $n = 20, 40, 60, 80, 100$ observations and a replication sample size of $m = n$ are shown in Table \ref{slr_table}. As seen in Table \ref{slr_table}, the confidence intervals achieve their nominal coverage probability.

\begin{table}[H]
\centering
\caption{Empirical coverage probability for $1 - \Phi(\sqrt{m}(c - 2) / \sqrt{25n(\mat{X}^\tran \mat{X})^{-1}_{2,2}}$ for different samples sizes $n = m$ and cutoffs $c$ for the slope of a simple linear regression. $K = 10,000$ datasets simulated for each sample size $n$.}
\begin{tabular}{cccccc}
\hline \hline
 & \multicolumn{5}{c}{$n$} \\
 \cline{2-6}
Cutoff & 20    & 40    & 60    & 80  & 100 \\
\hline
1.0 & 0.954 & 0.950 & 0.951 & 0.947 & 0.949 \\
1.1 & 0.954 & 0.950 & 0.951 & 0.947 & 0.949 \\
1.2 & 0.953 & 0.950 & 0.951 & 0.947 & 0.949 \\
1.3 & 0.952 & 0.950 & 0.951 & 0.947 & 0.949 \\
1.4 & 0.953 & 0.949 & 0.950 & 0.947 & 0.950 \\
1.5 & 0.952 & 0.950 & 0.950 & 0.947 & 0.951 \\
1.6 & 0.952 & 0.950 & 0.950 & 0.947 & 0.951 \\
1.7 & 0.953 & 0.949 & 0.950 & 0.948 & 0.950 \\
1.8 & 0.952 & 0.949 & 0.950 & 0.949 & 0.949 \\
1.9 & 0.953 & 0.950 & 0.951 & 0.949 & 0.950 \\
2.0 & 0.952 & 0.950 & 0.951 & 0.948 & 0.950 \\
2.1 & 0.953 & 0.950 & 0.951 & 0.948 & 0.950 \\
2.2 & 0.952 & 0.951 & 0.951 & 0.948 & 0.949 \\
2.3 & 0.952 & 0.951 & 0.951 & 0.948 & 0.949 \\
2.4 & 0.952 & 0.952 & 0.951 & 0.948 & 0.949 \\
2.5 & 0.951 & 0.951 & 0.951 & 0.948 & 0.949 \\
2.6 & 0.950 & 0.951 & 0.952 & 0.949 & 0.949 \\
2.7 & 0.950 & 0.950 & 0.953 & 0.949 & 0.950 \\
2.8 & 0.950 & 0.950 & 0.953 & 0.949 & 0.950 \\
2.9 & 0.952 & 0.951 & 0.953 & 0.950 & 0.949 \\
3.0 & 0.950 & 0.951 & 0.952 & 0.949 & 0.949
\end{tabular}
\label{slr_table}
\end{table}

\end{appendices}

\bibliographystyle{apalike}
\bibliography{refs}

\begin{thebibliography}{}

\bibitem[Amrhein and Greenland, 2018]{amrhein2018remove}
Amrhein, V. and Greenland, S. (2018).
\newblock Remove, rather than redefine, statistical significance.
\newblock {\em Nature Human Behaviour}, 2(1):4.

\bibitem[Amrhein et~al., 2019a]{amrhein2019retire}
Amrhein, V., Greenland, S., and McShane, B. (2019a).
\newblock Retire statistical significance.
\newblock {\em Nature}, 567:305--307.

\bibitem[Amrhein et~al., 2017]{amrhein2017earth}
Amrhein, V., Korner-Nievergelt, F., and Roth, T. (2017).
\newblock The earth is flat ($p> 0.05$): Significance thresholds and the crisis
  of unreplicable research.
\newblock {\em PeerJ}, 5:e3544.

\bibitem[Amrhein et~al., 2019b]{amrhein2019inferential}
Amrhein, V., Trafimow, D., and Greenland, S. (2019b).
\newblock Inferential statistics as descriptive statistics: There is no
  replication crisis if we don’t expect replication.
\newblock {\em The American Statistician}, 73(sup1):262--270.

\bibitem[Begley and Ellis, 2012]{begley2012}
Begley, C.~G. and Ellis, L.~M. (2012).
\newblock Drug development: Raise standards for preclinical cancer research.
\newblock {\em Nature}, 483:531--533.

\bibitem[Benjamin et~al., 2017]{benjamin2017}
Benjamin, D.~J., Berger, J.~O., Johannesson, M., Nosek, B.~A., Wagenmakers,
  E.-J., Berk, R., Bollen, K.~A., Brembs, B., Brown, L., Camerer, C., et~al.
  (2017).
\newblock Redefine statistical significance.
\newblock {\em Nature Human Behaviour}.

\bibitem[Berger and Mortera, 1991]{berger1991}
Berger, J.~O. and Mortera, J. (1991).
\newblock Interpreting the stars in precise hypothesis testing.
\newblock {\em International Statistical Review/Revue Internationale de
  Statistique}, 59(3):337--353.

\bibitem[Berger and Sellke, 1987]{berger1987}
Berger, J.~O. and Sellke, T. (1987).
\newblock Testing a point null hypothesis: The irreconcilability of p values
  and evidence.
\newblock {\em Journal of the American statistical Association},
  82(397):112--122.

\bibitem[Berry et~al., 2008]{berry2008}
Berry, C.~J., Shanks, D.~R., and Henson, R.~N. (2008).
\newblock A single-system account of the relationship between priming,
  recognition, and fluency.
\newblock {\em Journal of Experimental Psychology: Learning, Memory, and
  Cognition}, 34(1):97.

\bibitem[Billheimer, 2019]{billheimer2019}
Billheimer, D. (2019).
\newblock Predictive inference and scientific reproducibility.
\newblock {\em The American Statistician}, 73(sup1):291--295.

\bibitem[Boos and Stefanski, 2011]{boos2011}
Boos, D.~D. and Stefanski, L.~A. (2011).
\newblock P-value precision and reproducibility.
\newblock {\em The American Statistician}, 65(4):213--221.

\bibitem[Brandt et~al., 2014]{brandt2014}
Brandt, M.~J., IJzerman, H., Dijksterhuis, A., Farach, F.~J., Geller, J.,
  Giner-Sorolla, R., Grange, J.~A., Perugini, M., Spies, J.~R., and Van't~Veer,
  A. (2014).
\newblock The replication recipe: What makes for a convincing replication?
\newblock {\em Journal of Experimental Social Psychology}, 50:217--224.

\bibitem[Cassella and Berger, 1987]{casella1987}
Cassella, G. and Berger, J.~O. (1987).
\newblock Reconciling bayesian and frequentist evidence in the one-sided
  testing problem.
\newblock {\em Journal of the American statistical Association},
  82(397):106--111.

\bibitem[Cohen, 1994]{cohen1994}
Cohen, J. (1994).
\newblock The earth is round ($p<.05$).
\newblock {\em American Psychologist}, 49(12):997--1003.

\bibitem[Cumming and Maillardet, 2006]{cumming2006}
Cumming, G. and Maillardet, R. (2006).
\newblock Confidence intervals and replication: Where will the next mean fall?
\newblock {\em Psychological Methods}, 11(3):217.

\bibitem[De~Martini, 2008]{deMartini2008}
De~Martini, D. (2008).
\newblock Reproducibility probability estimation for testing statistical
  hypotheses.
\newblock {\em Statistics \& Probability Letters}, 78(9):1056--1061.

\bibitem[De~Martini, 2012]{deMartini2012}
De~Martini, D. (2012).
\newblock Stability criteria for the outcomes of statistical tests to assess
  drug effectiveness with a single study.
\newblock {\em Pharmaceutical Statistics}, 11(4):273--279.

\bibitem[De~Martini, 2013]{deMartini2013}
De~Martini, D. (2013).
\newblock {\em Success probability estimation with applications to clinical
  trials}.
\newblock John Wiley \& Sons.

\bibitem[{DeGroot}, 1973]{degroot1973}
{DeGroot}, M.~H. (1973).
\newblock Doing what comes naturally: Interpreting a tail area as a posterior
  probability or as a likelihood ratio.
\newblock {\em Journal of the American Statistical Association},
  68(344):966--969.

\bibitem[Dickey, 1977]{dickey1977}
Dickey, J.~M. (1977).
\newblock Is the tail area useful as an approximate bayes factor?
\newblock {\em Journal of the American Statistical Association},
  72(357):138--142.

\bibitem[Edwards et~al., 1963]{edwards1963}
Edwards, W., Lindman, H., and Savage, L.~J. (1963).
\newblock Bayesian statistical inference for psychological research.
\newblock {\em Psychological Review}, 70(3):193--242.

\bibitem[Fabrigar and Wegener, 2016]{fabrigar2016}
Fabrigar, L.~R. and Wegener, D.~T. (2016).
\newblock Conceptualizing and evaluating the replication of research results.
\newblock {\em Journal of Experimental Social Psychology}, 66:68--80.

\bibitem[Fricker~Jr et~al., 2019]{fricker2019}
Fricker~Jr, R.~D., Burke, K., Han, X., and Woodall, W.~H. (2019).
\newblock Assessing the statistical analyses used in basic and applied social
  psychology after their p-value ban.
\newblock {\em The American Statistician}, 73(sup1):374--384.

\bibitem[Gelman, 2015]{gelman2015}
Gelman, A. (2015).
\newblock The connection between varying treatment effects and the crisis of
  unreplicable research: A bayesian perspective.
\newblock {\em Journal of Management}, 41(2).

\bibitem[Gelman and Carlin, 2014]{gelman2014}
Gelman, A. and Carlin, J. (2014).
\newblock Beyond power calculations: Assessing type {S} (sign) and type {M}
  (magnitude) errors.
\newblock {\em Perspectives on Psychological Science}, 9(6):641--651.

\bibitem[Gelman et~al., 2014]{bda2014}
Gelman, A., Carlin, J.~B., Stern, H.~S., Dunson, D.~B., Vehatari, A., and
  Rubin, D.~B. (2014).
\newblock {\em Bayesian Data Analysis}.
\newblock Taylor \& Francis, Boca Raton, FL, 3rd edition.

\bibitem[Gigerenzer, 2004]{gigerenzer2004}
Gigerenzer, G. (2004).
\newblock Mindless statistics.
\newblock {\em The Journal of Socio-Economics}, 33:587--606.

\bibitem[Gigerenzer, 2018]{gigerenzer2018}
Gigerenzer, G. (2018).
\newblock Statistical rituals: The replication delusion and how we got there.
\newblock {\em Advances in Methods and Practices in Psychological Science},
  1(2):198--218.

\bibitem[Goodman, 1992]{goodman1992}
Goodman, S.~N. (1992).
\newblock A comment on replication, p-values and evidence.
\newblock {\em Statistics in Medicine}, 11(7):875--879.

\bibitem[Greenland, 2017]{greenland2017}
Greenland, S. (2017).
\newblock Invited commentary: The need for cognitive science in methodology.
\newblock {\em American Journal of Epidemiology}, 186(6):639--645.

\bibitem[Haller and Krauss, 2002]{haller2002}
Haller, H. and Krauss, S. (2002).
\newblock Misinterpretations of significance: A problem students share with
  their teachers?
\newblock {\em Methods of Psychological Research}, 7(1):1--20.

\bibitem[Ioannidis, 2005]{ioannidis2005}
Ioannidis, J. P.~A. (2005).
\newblock Why most published research findings are false.
\newblock {\em PLoS Medicine}, 2(8):e124.

\bibitem[Jeffreys, 1935]{jeffreys1935}
Jeffreys, H. (1935).
\newblock Some tests of significance, treated by the theory of probability.
\newblock {\em Proceedings of the Cambridge Philosophical Society},
  31(2):203--222.

\bibitem[Jeffreys, 1961]{jeffreys1961}
Jeffreys, H. (1961).
\newblock {\em Theory of Probability}.
\newblock Oxford university press, 3rd edition.

\bibitem[Johnson et~al., 2017]{johnson2017}
Johnson, V.~E., Payne, R.~D., Wang, T., Asher, A., and Mandal, S. (2017).
\newblock On the reproducibility of psychological science.
\newblock {\em Journal of the American Statistical Association},
  112(517):1--10.

\bibitem[Kass and Raftery, 1995]{kass1995}
Kass, R.~E. and Raftery, A.~E. (1995).
\newblock Bayes factors.
\newblock {\em Journal of the American Statistical Association},
  90(430):773--795.

\bibitem[Kennedy-Shaffer, 2019]{kennedy2019}
Kennedy-Shaffer, L. (2019).
\newblock Before $p < 0.05$ to beyond $p < 0.05$: Using history to
  contextualize p-values and significance testing.
\newblock {\em The American Statistician}, 73(sup1):82--90.

\bibitem[K{\"u}hberger et~al., 2014]{kuhberger2014}
K{\"u}hberger, A., Fritz, A., and Scherndl, T. (2014).
\newblock Publication bias in psychology: A diagnosis based on the correlation
  between effect size and sample size.
\newblock {\em PloS one}, 9(9):e105825.

\bibitem[Lakens et~al., 2018]{lakens2018}
Lakens, D., Adolfi, F.~G., Albers, C.~J., Anvari, F., Apps, M. A.~J., Argamon,
  S.~E., Baguley, T., Becker, R.~B., Benning, S.~D., Bradford, D.~E., et~al.
  (2018).
\newblock Justify your alpha.
\newblock {\em Nature Human Behaviour}, 2(3):168--171.

\bibitem[Lambert and Hall, 1982]{lambert1982}
Lambert, D. and Hall, W.~J. (1982).
\newblock Asymptotic lognormality of p-values.
\newblock {\em The Annals of Statistics}, 10(1):44--64.

\bibitem[Leek et~al., 2017]{leek2017}
Leek, J., McShane, B.~B., Gelman, A., Colquhoun, D., Nuijten, M.~B., and
  Goodman, S.~N. (2017).
\newblock Five ways to fix statistics.
\newblock {\em Nature}, 551:557--559.

\bibitem[Lindley, 1957]{lindley1957}
Lindley, D.~V. (1957).
\newblock A statistical paradox.
\newblock {\em Biometrika}, 44(1/2):187--192.

\bibitem[McCloskey and Ziliak, 1996]{mccloskey1996}
McCloskey, D.~N. and Ziliak, S.~T. (1996).
\newblock The standard error of regressions.
\newblock {\em Journal of Economic Literature}, 34:97--114.

\bibitem[McShane and B{\"o}ckenholt, 2014]{mcshane2014}
McShane, B.~B. and B{\"o}ckenholt, U. (2014).
\newblock You cannot step into the same river twice: When power analyses are
  optimistic.
\newblock {\em Perspectives on Psychological Science}, 9(6):612--625.

\bibitem[McShane and Gal, 2016]{mcshane2016}
McShane, B.~B. and Gal, D. (2016).
\newblock Blinding us to the obvious? {T}he effect of statistical training on
  the evaluation of evidence.
\newblock {\em Management Science}, 62(6):1707--1718.

\bibitem[McShane and Gal, 2017]{mcshane2017}
McShane, B.~B. and Gal, D. (2017).
\newblock Statistical significance and the dichotomization of evidence.
\newblock {\em Journal of the American Statistical Association},
  112(519):885--895.

\bibitem[McShane et~al., 2019a]{mcshane2019}
McShane, B.~B., Gal, D., Gelman, A., Robert, C., and Tackett, J.~L. (2019a).
\newblock Abandon statistical significance.
\newblock {\em The American Statistician}, 73(sup1):235--245.

\bibitem[McShane et~al., 2019b]{mcshane2019large}
McShane, B.~B., Tackett, J.~L., B{\"o}ckenholt, U., and Gelman, A. (2019b).
\newblock Large-scale replication projects in contemporary psychological
  research.
\newblock {\em The American Statistician}, 73(sup1):99--105.

\bibitem[Meeker et~al., 2017]{meeker2017}
Meeker, W.~Q., Hahn, G.~J., and Escobar, L.~A. (2017).
\newblock {\em Statistical intervals: A guide for practitioners and
  researchers}.
\newblock John Wiley \& Sons, Hoboken, NJ, 2nd edition.

\bibitem[Meixner and Bruening, 2015]{meixner2015}
Meixner, J. and Bruening, J. (2015).
\newblock {Replication of ``A single-system account of the relationship between
  priming, recognition, and fluency'' by C. J. Berry, D. R. Shanks, R. N.
  Henson (2008, Journal of Experimental Psychology: Learning, Memory, and
  Cognition)}.
\newblock Technical report.

\bibitem[Morey and Rouder, 2011]{morey2011}
Morey, R.~D. and Rouder, J.~N. (2011).
\newblock Bayes factor approaches for testing interval null hypotheses.
\newblock {\em Psychological Methods}, 16(4):406--419.

\bibitem[Morey and Rouder, 2015]{BayesFactor}
Morey, R.~D. and Rouder, J.~N. (2015).
\newblock {\em BayesFactor: Computation of Bayes factors for common designs}.
\newblock R package version 0.9.12-2.

\bibitem[Nuzzo, 2014]{nuzzo2014}
Nuzzo, R. (2014).
\newblock Scientific method: Statistical errors.
\newblock {\em Nature}, 506(7487):150--152.

\bibitem[Oakes, 1986]{oakes1986}
Oakes, M. (1986).
\newblock {\em Statistical inference: A commentary for the social and
  behavioural sciences}.
\newblock John Wiley \& Sons.

\bibitem[{Open Science Collaboration}, 2015]{open2015}
{Open Science Collaboration} (2015).
\newblock Estimating the reproducibility of psychological science.
\newblock {\em Science}, 349(6251):943, aac4716.

\bibitem[Pashler and Wagenmakers, 2012]{pashler2012}
Pashler, H. and Wagenmakers, E.-J. (2012).
\newblock Editors’ introduction to the special section on replicability in
  psychological science: A crisis of confidence?
\newblock {\em Perspectives on Psychological Science}, 7(6):528--530.

\bibitem[Pratt, 1965]{pratt1965}
Pratt, J.~W. (1965).
\newblock Bayesian interpretation of standard inference statements.
\newblock {\em Journal of the Royal Statistical Society, Series B
  (Methodological)}, 27(2):169--203.

\bibitem[Prinz et~al., 2011]{prinz2011}
Prinz, F., Schlange, T., and Asadullah, K. (2011).
\newblock Believe it or not: How much can we rely on published data on
  potential drug targets?
\newblock {\em Nature Reviews Drug Discovery}, 10(9):712.

\bibitem[Rosenthal, 1997]{rosenthal1997}
Rosenthal, R. (1997).
\newblock Some issues in the replication of social science research.
\newblock {\em Labour Economics}, 4:121--123.

\bibitem[Rouder et~al., 2009]{rouder2009}
Rouder, J.~N., Speckman, P.~L., Sun, D., Morey, R.~D., and Iverson, G. (2009).
\newblock Bayesian t tests for accepting and rejecting the null hypothesis.
\newblock {\em Psychonomic Bulletin \& Review}, 16(2):225--237.

\bibitem[Ruberg et~al., 2019]{ruberg2019}
Ruberg, S.~J., Harrell~Jr, F.~E., Gamalo-Siebers, M., LaVange, L., Jack~Lee,
  J., Price, K., and Peck, C. (2019).
\newblock Inference and decision making for 21st-century drug development and
  approval.
\newblock {\em The American Statistician}, 73(sup1):319--327.

\bibitem[Schweder and Hjort, 2016]{schweder2016}
Schweder, T. and Hjort, N.~L. (2016).
\newblock {\em Confidence, Likelihood, Probability: Statistical Inference with
  Confidence Distributions}.
\newblock Cambridge University Press, 1st edition.

\bibitem[Shafer, 1982]{shafer1982}
Shafer, G. (1982).
\newblock Lindley's paradox.
\newblock {\em Journal of the American Statistical Association},
  77(378):325--334.

\bibitem[Shao, 2003]{shao2003}
Shao, J. (2003).
\newblock {\em Mathematical Statistics}.
\newblock Springer, 2nd edition.

\bibitem[Shao and Chow, 2002]{shao2002}
Shao, J. and Chow, S.-C. (2002).
\newblock Reproducibility probability in clinical trials.
\newblock {\em Statistics in Medicine}, 21(12):1727--1742.

\bibitem[Simons et~al., 2017]{simons2017}
Simons, D.~J., Shoda, Y., and Lindsay, D.~S. (2017).
\newblock Constraints on generality ({COG}): A proposed addition to all
  empirical papers.
\newblock {\em Perspectives on Psychological Science}, 12(6):1123--1128.

\bibitem[Stanley et~al., 2018]{stanley2018}
Stanley, T.~D., Carter, E.~C., and Doucouliagos, H. (2018).
\newblock What meta-analyses reveal about the replicability of psychological
  research.
\newblock {\em Psychological Bulletin}, 144(12):1325--1346.

\bibitem[Stroebe and Strack, 2014]{stroebe2014}
Stroebe, W. and Strack, F. (2014).
\newblock The alleged crisis and the illusion of exact replication.
\newblock {\em Perspectives on Psychological Science}, 9(1):59--71.

\bibitem[Trafimow et~al., 2018]{trafimow2018}
Trafimow, D., Amrhein, V., Areshenkoff, C.~N., Barrera-Causil, C.~J., Beh,
  E.~J., Bilgi{\c{c}}, Y.~K., Bono, R., Bradley, M.~T., Briggs, W.~M.,
  Cepeda-Freyre, H.~A., et~al. (2018).
\newblock Manipulating the alpha level cannot cure significance testing.
\newblock {\em Frontiers in Psychology}, 9:1--7.

\bibitem[Van~Erp et~al., 2017]{vanerp2017}
Van~Erp, S., Verhagen, J., Grasman, R. P. P.~P., and Wagenmakers, E.-J. (2017).
\newblock Estimates of between-study heterogeneity for 705 meta-analyses
  reported in psychological bulletin from 1990--2013.
\newblock {\em Journal of Open Psychology Data}, 5(1).

\bibitem[Wasserstein and Lazar, 2016]{wasserstein2016}
Wasserstein, R.~L. and Lazar, N.~A. (2016).
\newblock The {ASA}'s statement on p-values: Context, process, and purpose.
\newblock {\em The American Statistician}, 70(2):129--133.

\bibitem[Wasserstein et~al., 2019]{wasserstein2019}
Wasserstein, R.~L., Schirm, A.~L., and Lazar, N.~A. (2019).
\newblock Moving to a world beyond ``$p < 0.05$''.
\newblock {\em The American Statistician}, 73(sup1):1--19.

\end{thebibliography}

\end{document}